\documentclass[12pt]{amsart}

\usepackage{latexsym,amssymb,amsfonts,amsmath,mathrsfs,float,amsthm}
\usepackage{xcolor}
\usepackage{graphicx}
\usepackage[linktocpage=true, linkbordercolor=red]{hyperref}

  \newcommand{\R}{\mathbb{R}} 
  \newcommand{\C}{\mathbb{C}} 
  \newcommand{\N}{\mathbb{N}} 
  \newcommand{\Z}{\mathbb{Z}} 
  \newcommand{\pmset}[1]{\{-1,1\}^{#1}} 
  \newcommand{\bset}[1]{\{0,1\}^{#1}} 
  
  \newcommand{\st}{:\,} 
  
  \newcommand{\ie}{{i.e.}}

\newcommand{\condeq}[1]{
	\left\{\begin{array}{ll}
	#1
	\end{array}\right.
}

  \newcommand{\F}{\mathbb{F}}

  \newcommand{\ceil}[1]{\lceil{#1}\rceil}
  
  \DeclareMathOperator{\Tr}{\mathsf{Tr}}
  \renewcommand{\Pr}{\mbox{\rm Pr}}
  
  \DeclareMathOperator{\Diag}{Diag}

\newcommand{\beq}{\begin{equation}}
\newcommand{\eeq}{\end{equation}}
\newcommand{\beqn}{\begin{equation*}}
\newcommand{\eeqn}{\end{equation*}}
\newcommand{\beqr}{\begin{eqnarray}}
\newcommand{\eeqr}{\end{eqnarray}}
\newcommand{\beqrn}{\begin{eqnarray*}}
\newcommand{\eeqrn}{\end{eqnarray*}}
\newcommand{\bmline}{\begin{multline}}
\newcommand{\emline}{\end{multline}}
\newcommand{\bmlinen}{\begin{multline*}}
\newcommand{\emlinen}{\end{multline*}}

\newtheorem{defin}{Definition}[section]
\newtheorem{definition}[defin]{Definition}
\newtheorem{proposition}[defin]{Proposition}
\newtheorem{theorem}[defin]{Theorem}

\newtheorem{corollary}[defin]{Corollary}
\newtheorem{lemma}[defin]{Lemma}

\makeatletter
\newtheorem*{rep@theorem}{\rep@title}
\newcommand{\newreptheorem}[2]{%
\newenvironment{rep#1}[1]{%
 \def\rep@title{#2 \ref{##1}}%
 \begin{rep@theorem}}%
 {\end{rep@theorem}}}
\makeatother
\newreptheorem{theorem}{Theorem}

\usepackage{enumitem}
\makeatletter
\def\namedlabel#1#2{\begingroup
    #2%
    \def\@currentlabel{#2}%
    \phantomsection\label{#1}\endgroup
}

\usepackage[font=footnotesize,labelfont=bf, width=1\textwidth]{caption} 
\usepackage{fullpage}
\newcommand{\qindep}{\alpha^{\star}} 
\newcommand{\qchrom}{\chi^{\star}} 
\newcommand{\qwits}{R^{\star}} 
\newcommand{\scrate}{\eta}
\newcommand{\qscrate}{\eta^{\star}}
\newcommand{\qcap}{c^{\star}}

\newcommand{\thplus}{\vartheta^+}

\newcommand{\ignore}[1]{}

\begin{document}


\title[Entanglement-assisted zero-error source-channel coding]{\bf \large Entanglement-assisted zero-error source-channel coding}
\thanks{
J.~B.\ was partially supported by a Rubicon grant from the Netherlands Organization for Scientific Research (NWO).
J.~B.\ and H.~B.\ were supported by the European Commission under the project QCS (Grant No. 255961). 
H.~B.\ and T.~P.\ were supported by the EU grant SIQS.
G.~S.\ was supported by Ronald de Wolf's Vidi grant 639.072.803 from the Netherlands Organization for Scientific Research (NWO)}

\author[J. Bri\"{e}t]{Jop Bri\"et}
\address{Courant Institute, New York University, 251 Mercer Street, New York NY 10012, USA}
\email{jop.briet@cims.nyu.edu}

\author[H. Buhrman]{Harry Buhrman}
\address{Centrum Wiskunde \& Informatica (CWI), Science Park 123, 1098 XG Amsterdam, The Netherlands}
\email{buhrman@cwi.nl}

\author[M. Laurent]{Monique Laurent}
\address{Centrum Wiskunde \& Informatica (CWI), Science Park 123, 1098 XG Amsterdam,  and   Tilburg University,
              PO Box 90153, 5000 LE Tilburg, The Netherlands}
              \email{m.laurent@cwi.nl}

\author[T. Piovesan]{Teresa Piovesan}
\address{Centrum Wiskunde \& Informatica (CWI), Science Park 123, 1098 XG Amsterdam, The Netherlands}
\email{t.piovesan@cwi.nl}

\author[G. Scarpa]{Giannicola Scarpa}
\address{Universitat Aut\`onoma de Barcelona, E-08193 Bellaterra (Barcelona), Spain}
\email{giannicolascarpa@gmail.com}

\maketitle

\begin{abstract}
We study the use of quantum entanglement in the zero-error source-channel coding problem.
Here, Alice and Bob are connected by a noisy classical one-way channel, and are given correlated inputs from a random source.
Their goal is for Bob to learn Alice's input while using the channel as little as possible.
In the zero-error regime, the optimal rates of source codes and channel codes are given by graph parameters known as the Witsenhausen rate and Shannon capacity, respectively.
The Lov\'asz theta number, a graph parameter defined by a semidefinite program, gives the best efficiently-computable upper bound on the Shannon capacity and it also upper bounds its entanglement-assisted counterpart.
At the same time it was recently shown that the Shannon capacity can be increased if Alice and Bob may use entanglement.

Here we partially extend these results to the source-coding problem and to the more general source-channel coding problem.
We prove a lower bound on the rate of entanglement-assisted source-codes in terms of Szegedy's number (a strengthening of the theta number).
This result implies that the theta number  lower bounds the entangled variant of the Witsenhausen rate.
We also show that entanglement can allow for an unbounded improvement of the asymptotic rate of both classical source codes and classical source-channel codes.
Our separation results use low-degree polynomials due to Barrington, Beigel and Rudich, Hadamard matrices due to Xia and Liu and a new application of remote state preparation.
\end{abstract}


%
%

%
\section{Introduction}\label{sec:intro}

We study a problem from classical zero-error information theory:  the {\em zero-error source-channel coding problem}, in the setting where the sender and receiver may use quantum entanglement.
Viewed separately, the source coding problem asks a sender, Alice, to efficiently communicate data about which a receiver, Bob, already has some information\footnote{This setting is usually called the {\em dual} source coding problem, referring to the fact that it involves two parties. Since we only consider the two-party setting we omit the word {\em dual}.}, while the channel coding problem asks Alice to transmit data reliably in the presence of noise.
In the combination of these two problems, Alice and Bob are each given an input from a random source and get access to a noisy channel through which Alice can send messages to Bob.
Their goal is to minimize the average number of channel uses per source input such that Bob can learn Alice's inputs with zero probability of error.

Shannon's seminal paper~\cite{Shannon:1956} on zero-error channel capacity kindled a large research area which involves not only information theorists but also researchers from combinatorics, computer science and mathematical programming (see for example K\"{o}rner and Orlitsky~\cite{Korner:1998} for an extensive survey and Lubetzky's PhD thesis~\cite{Lubetzky:2007} for more recent results).
The branch of this line of research involving entanglement was started only recently by Cubitt et al.~\cite{Cubitt:2010}.
The possibility for a pair of quantum systems to be entangled is one of the most striking features of quantum mechanics.
The typical setting in which this phenomenon manifests itself is where two parties, Alice and Bob, each have a quantum system and perform on it a measurement of their choice.
In a celebrated response to a paper of Einstein, Podolsky and Rosen~\cite{Einstein:1935}, Bell~\cite{Bell:1964} showed that entanglement between Alice's and Bob's systems can cause their measurement outcomes to be distributed according to probability distributions that fall outside the realm of classical physics.
In particular, entanglement can give outcome pairs which do not follow a product distribution, nor any convex combination of such distributions.
Entanglement therefore allows spatially separated parties to produce so-called {\em non-local} correlations without needing to communicate.
Our main results concern {\em lower bounds} on the optimum rate of entanglement-assisted source codes and the {\em advantage} that entanglement can give in the source-channel coding problem. 
Below we set the stage in detail and state our results precisely.

First, we recall some preliminary definitions needed for this paper.
Throughout, all graphs are assumed to be finite, undirected and without self-loops.
For a graph $G = (V,E)$, $V$ and $E$ (or equivalently $V(G)$ and $E(G)$) denote its vertex and edge set, respectively.
The complement of $G$ is $\overline G$, the graph with vertex set $V(G)$ where distinct vertices are adjacent if and only if they are not adjacent in $G$.
An independent set is a subset of the vertex set such that no pair is adjacent and the {\em independence number}~$\alpha(G)$ is the maximum cardinality of an independent set  in $G$.
A clique is a subset of vertices in which each pair is adjacent and the {\em clique number} $\omega(G)$ is the maximum cardinality of a clique in $G$. 
Clearly $\alpha(G) = \omega(\overline G)$.
A {\em proper coloring} is a set of pairwise disjoint independence sets that cover $V(G)$, \ie,~an assignment of a color to each vertex such that adjacent vertices receive distinct colors.
The {\em chromatic number}~$\chi(G)$ is the minimum number of colors needed for a proper coloring.
Thus $\alpha(G) \cdot \chi(G) \geq |V(G)|$.
The {\em strong product} $G \boxtimes H$ of two graphs $G$ and $H$ is the graph whose vertex set is the cartesian product $V(G) \times V(H)$ and where two distinct vertices $(u_1,u_2)$ and $(v_1,v_2)$ are adjacent if and only if it holds that $u_1=v_1$ or $\{u_1,v_1\} \in E(G)$ and that $u_2=v_2$ or $\{u_2,v_2\} \in E(H)$. 
For a graph~$G$ and $m\in \N$, $G^{\boxtimes m}$ denotes the strong product of $m$ copies of $G$, with vertex set $V(G)^m$ and where two distinct vertices 
$(u_1,\ldots,u_m)$ and $(v_1,\ldots,v_m)$ are adjacent if, for all $i\in [m]$, either $u_i=v_i$ or $\{u_i,v_i\}\in E(G)$.
A {\em homomorphism} from a graph $G$ to a graph $H$ is a map $\phi:V(G)\to V(H)$ such that every edge $\{u,v\}$ in $G$ is mapped to an edge $\{\phi(u),\phi(v)\}$ in~$H$. 
If such a map exists, we write $G\longrightarrow H$. 
The complete graph on $t$ vertices, denoted by~$K_t$, is the graph where every pair of distinct vertices is adjacent.
A {\em $d$-dimensional orthogonal representation} of a graph $G$ is a map $f$ from $V(G)$ to non-zero vectors in~$\C^d$ such that adjacent vertices are mapped to orthogonal vectors.\footnote{We stress that in our definition {\em orthogonality corresponds to adjacency}. Some authors prefer to demand orthogonality for non-adjacent vertices instead.}
The {\em orthogonal rank~$\xi(G)$} of $G$ is the minimum~$d$ such that there exists a $d$-dimensional orthogonal representation of~$G$.
Following~\cite{Cameron:2007} we define $\xi'(G)$ to be the minimum~$d$ for which there exists a $d$-dimensional orthogonal representation~$f$ of~$G$ such that for each~$u\in V(G)$ the entries of the vector~$f(u)$ all have absolute value one.
To indicate that a Hermitian matrix $A$ is positive semidefinite we write $A \succeq 0$.
Recall that~$A\succeq 0$ if and only if there exist vectors~$a_i$ such that~$A_{ij} = a_i^*a_j$.
The trace inner product of matrices $A,B$ of equal size is defined by $\langle A,B \rangle = \Tr(A^*B)$. 
Finally, for $n\in\N$ we denote  $[n] = \{1,2,\dots,n\}$.

\subsection{Classical source-channel coding}
In this section we describe the classical zero-error source, channel, and source-channel coding problems.
A {\em source} ${\mathcal{M}= (\mathsf X, \mathsf U, P)}$ consists of a finite set~$\mathsf X$, a (possibly infinite) set~$\mathsf U$ and a probability distribution~$P$ over~$\mathsf X\times\mathsf U$.
In a source instance, Alice is given an input $x\in \mathsf X$ and Bob an input $u\in \mathsf U$ with probability~$P(x,u)$.
Bob's input may already give him some information about Alice's. 
But if his input does not uniquely identify hers, she has to supply additional information for him to learn it exactly.
For this they get access to a noiseless one-way binary channel which they aim to use as little as possible.\footnote{From now on we will assume that all binary channels are noiseless.}
Here we consider only {\em memoryless} sources, which means that the probability distribution $P(x,u)$ of the source is unchanged after every instance.

The source-coding problem can sometimes be solved more efficiently by jointly encoding sequences of inputs into single codewords.
If the parties use {\em block codes} of length $n$ to deal with length $m$ input sequences, then after receiving an input sequence $\mathbf x = (x_1,\dots,x_m)$, Alice applies encoding function $\mathsf C:\mathsf X^m\to\bset{n}$ and sends $\mathsf C(\mathbf x)$ through the binary channel by using it $n$ times in a row.
Bob, who received an input $\mathbf{u} = (u_1,\dots,u_m)\in\mathsf U^m$, then applies a decoding function $\mathsf D:\mathsf U^m\times \bset{n} \to \mathsf X^m$ to the pair $(\mathbf u,\mathsf C(\mathbf x))$ to get a string in~$\mathsf X^m$.
The scheme works if Bob always gets the string~$\mathbf x$.
The {\em cost rate} of the scheme $(\mathsf C,\mathsf D)$ is then~$n/m$, which counts the average number of channel uses per source-input symbol.

Witsenhausen~\cite{Witsenhausen:1976} and Ferguson and Bailey~\cite{Ferguson:1975}  showed that the zero-error source coding problem can be studied in graph-theoretic terms.
Associated with a source $\mathcal M = (\mathsf X,\mathsf U,P)$ is its {\em characteristic graph} $G = (\mathsf X, E)$, where $\{x,y\}\in E$ if there exists a $u\in\mathsf U$ such that $P(x,u)>0$ and $P(y,u)>0$.
As such, the edge set identifies the pairs of inputs for Alice which Bob cannot distinguish based on his input.
It is not difficult to see that every graph is the characteristic graph of a (non-unique) source.
Solving one instance of the zero-error source coding problem for~$\mathcal M$ is equivalent to finding a proper coloring of~$G$.
Indeed, Bob's input~$u$ reduces the list of Alice's possible inputs to the set $\{x\in\mathsf X\st P(x,u)>0\}$ and this set forms a clique in~$G$.
So Bob can learn Alice's input if she sends him its color.
Conversely, a length $1$ block-code for~$\mathcal M$ defines a proper coloring of~$G$.
To deal with length $m$ input sequences we take the graph~$G^{\boxtimes m}$ (the strong product of~$m$ copies of $G$), whose edges are the pairs of input sequences for Alice which Bob cannot distinguish.
The {\em Witsenhausen rate} 
\beq\label{eq:witsdef}
R(G) = \lim_{m\to\infty}\frac{1}{m}\log\chi(G^{\boxtimes m})
\eeq
is the minimum asymptotic cost rate of a zero-error code for a source (all logarithms in this paper are in base~2). 
As is well known, the chromatic number is sub-multiplicative, \ie,
$\chi(G^{\boxtimes (m+m')}) \le \chi(G^{\boxtimes m})\chi(G^{\boxtimes m'})$.
Therefore, by Fekete's lemma\footnote{If a sequence $(a_m)_{m\in\N}$  is sub-additive (\ie, $a_{m+m'}\le a_m+a_{m'}$ for all $m,m'\in\N$), Fekete's lemma claims that  the sequence 
$\left( a_m/m \right)_{m\in\N}$ has a limit, which is equal to its infimum:  $\lim_{m\rightarrow \infty} a_m/m = \inf_{m\in\N} a_m/m$.} 
  the above limit exists  and  is equal to the infimum: $R(G)=\inf_{m}\log\chi(G^{\boxtimes m})/m$. 

A (noisy) {\em discrete channel} $\mathcal N = (\mathsf S,\mathsf V, Q)$ consists of a finite input set~$\mathsf S$, a (possibly infinite) output set~$\mathsf V$ and a probability distribution $Q(\cdot | s)$ over~$\mathsf V$ for each $s\in\mathsf S$.
Throughout the paper we consider only memoryless channels, where the probability distribution of the output depends only on the current channel input.
If Alice sends an input $s\in\mathsf S$ through the channel, then Bob receives the output $v\in\mathsf V$ with probability~$Q(v|s)$.
Their goal is to transmit a binary string~$\mathbf y$ of, say, $m$ bits from Alice to Bob while using the channel as little as possible.
If the parties use a block code of length~$n$, then Alice has an encoding function~${\mathsf C:\bset{m}\to\mathsf S^n}$ and sends $\mathsf C(\mathbf y)$ through the channel by using it~$n$ times in sequence.
Bob then receives an output sequence $\mathbf v = (v_1,\dots,v_n)$ on his side of the channel and applies a decoding function~${\mathsf D:\mathsf V^n\to\bset{m}}$.
The coding scheme $(\mathsf C,\mathsf D)$ works if $\mathsf D(\mathbf v) = \mathbf y$.
The {\em communication rate} of the scheme is~$m/n$, the number of bits transmitted per channel use.

Shannon~\cite{Shannon:1956} showed that the zero-error channel coding problem can also be studied in graph-theoretic terms.
Associated to a channel~$\mathcal N = (\mathsf S, \mathsf V, Q)$ is its {\em confusability graph} $H = (\mathsf S, F)$ where $\{s,t\}\in F$ if there exists a $v\in\mathsf V$ such that both $Q(v|s)>0$ and $Q(v|t)>0$.
The edge set identifies pairs which can lead to identical channel outputs on Bob's side.
Sets of non-confusable inputs thus correspond to independent sets in~$H$. 
Codes of block-length~$n$ then allow the zero-error transmission of~$\alpha(H^{\boxtimes n})$ distinct messages.
The {\em Shannon capacity}
\beq\label{eq:capdef}
c(H) = \lim_{n\to\infty}\frac{1}{n}\log\alpha(H^{\boxtimes n})
\eeq
is the maximum communication rate of a zero-error coding scheme.
As for the Witsenhausen rate, we can replace the above limit with the supremum:
$c(H)=\sup_n \log\alpha(H^{\boxtimes n})/n$.

In the {\em source-channel coding problem} the parties receive inputs from a source~$\mathcal M = (\mathsf X,\mathsf U,P)$ and get access to a channel~$\mathcal N = (\mathsf S,\mathsf V,Q)$.
Their goal is to solve the source coding problem, but now using the (noisy) channel~$\mathcal N$ instead of a (noiseless) binary channel.
An {\em $(m,n)$-coding scheme} for this problem consists of an encoding function~$\mathsf C:\mathsf X^m\to\mathsf S^n$ and a decoding function~$\mathsf D:\mathsf U^m\times\mathsf V^n\to\mathsf X^m$ (see Figure~\ref{fig:qwits-classic}).
The {\em cost rate} is~$n/m$.

\begin{figure}[t]
\begin{center}
 \includegraphics[width=8.87cm]{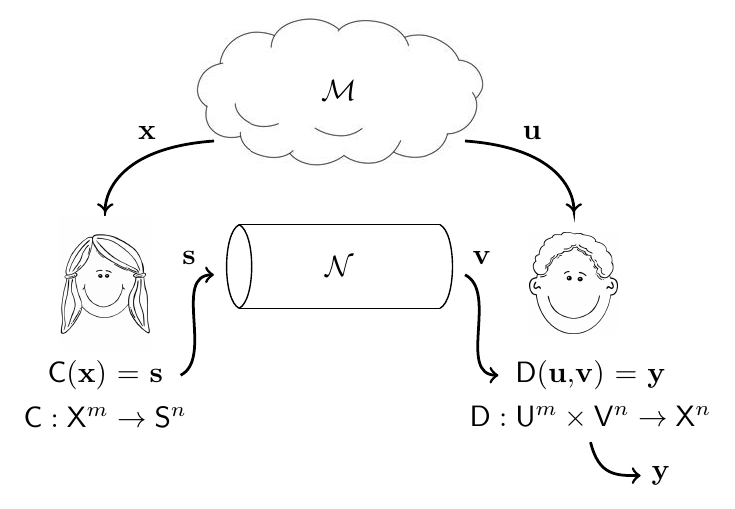}
 \caption{The figure illustrates a classical source-coding instance where Alice and Bob use an~$(m,n)$-coding scheme~$(\mathsf C,\mathsf D)$. The parties receive length $m$ input strings~$\mathbf x$ and $\mathbf u$, respectively, from a source~$\mathcal M = (\mathsf X, \mathsf U, P)$ and have a one-way channel~$\mathcal N = (\mathsf S, \mathsf V, Q)$. Using~$\mathsf C$, Alice encodes her input into a string~$\mathbf s\in \mathsf S^n$ which she sends through the channel. After receiving a channel output~$\mathbf v$, Bob applies~$\mathsf D$ to the pair~$(\mathbf u, \mathbf v)$ to get a string~$\mathbf y$. The scheme works if~$\mathbf y = \mathbf x$.}
\label{fig:qwits-classic}
\end{center}
\end{figure}

Nayak, Tuncel and Rose~\cite{Nayak:2006} showed that if~$\mathcal M$ has characteristic graph~$G$ and~$\mathcal N$ has confusability graph~$H$, then a zero-error $(m,n)$-coding scheme is equivalent to a homomorphism
 from~$G^{\boxtimes m}$ to~$\overline{H^{\boxtimes n}}$.
Then,
 the parameter
\begin{equation}\label{eqeta}
\eta(G,H):= \lim_{m\to\infty}\frac{1}{m}\min\Big\{n\in\N\st  G^{\boxtimes m} \longrightarrow  \overline{H^{\boxtimes n}}\Big\}
\end{equation}
gives the minimum asymptotic cost rate of a zero-error code. 
We will assume throughout that both $G$ and $\overline H$ contain at least one edge. (Indeed, if $G$ has no edge then $\eta(G,H)=0$ for any $H$ and, if $G$ has at least one edge, then $\eta(G,H)$ is well defined only if $\overline H$ has at least one edge.)
To see that the limit exists,  observe that the parameter
$$
\scrate_m(G,H):= \min\Big\{n\in\N\st  G^{\boxtimes m} \longrightarrow \overline{H^{\boxtimes n}}\Big\}
$$
is sub-additive and apply  Fekete's lemma, which shows that  $\eta(G,H)=\lim_{m\rightarrow \infty} \scrate_m(G,H)/m $ is also equal to the infimum $\inf_m \scrate_m(G,H)/m$.

If the channel $\mathcal N$ is replaced by a binary channel we regain the source coding problem.
Conversely, if Alice receives binary inputs from the source and Bob's source inputs give him no information about Alice's at all, then we regain the channel coding problem.
More formally, we can reformulate $R(G)$ and $c(H)$ in the following way.

\begin{lemma}\label{lemRc}
Let $G$ and $H$ be graphs such that both $G$ and $\overline H$  have at least one edge.
Then, 
$$R(G) = \scrate(G,\overline K_2) \ \text{ and } \ 
  1/c(H) = \scrate(K_2,H).$$
\end{lemma}

\begin{proof}
For the proof of the identity $R(G)=\scrate(G,\overline K_2)$ we use the following simple fact: for a graph $G'$ and ${t\in \N}$,   there exists a homomorphism from $G'$ to $K_{t}$ if and only if~${\chi(G') \leq t}$, which implies
\beqn
\log\chi(G') \leq \min\{n \st G' \longrightarrow K_{2^n}\} < \log\chi(G') +1.
\eeqn
Combining these inequalities applied to $G'=G^{\boxtimes m}$ with the identity  $\overline{ \overline{K}_2^{\boxtimes n}} = K_{2^n}$, we obtain  
\beqrn
\scrate(G,\overline K_2) 
&=& \lim_{m \to \infty} \frac{1}{m}\min\{n \st G^{\boxtimes m} \longrightarrow \overline{ \overline{K}_2^{\boxtimes n}} =K_{2^n}\} 
=
 \lim_{m \to \infty} \frac{1}{m}\log \chi(G^{\boxtimes m})
=R(G).
\eeqrn

The proof of the identity $1/c(H) = \scrate(K_2,H)$ uses  the fact that, for a graph~$H'$ and $t\in\N$, there exists a homomorphism from~$K_t$ to $\overline{H'}$ if and only if $\alpha(H') \geq t$.
Since~$K_2^{\boxtimes m} = K_{2^m}$, we get
\beqrn
\scrate_m(K_2,H) &=& \min\big\{n\st K_2^{\boxtimes m}=K_{2^m} \longrightarrow \overline{H^{\boxtimes n}}\big\}\\
&=& 
\min\big\{n \st \alpha(H^{\boxtimes n}) \geq 2^m\big\}
=
\min\big\{n \st  \log \alpha(H^{\boxtimes n})\geq m\big\}.
\eeqrn
Setting $n(m):=\scrate_m(K_2,H)$, this implies
$$
\log \alpha(H^{\boxtimes(n(m)-1)}) < m\leq \log \alpha(H^{\boxtimes n(m)})$$
and thus 
\begin{equation}\label{relc}
\frac{n(m)}{\log \alpha(H^{\boxtimes n(m)})} \leq \frac{n(m)}{m} < \frac{n(m)}{\log \alpha(H^{\boxtimes(n(m)-1)})}.
\end{equation}
As $c(H)=\sup_n \log \alpha(H^{\boxtimes n})/n$, using the left most inequality in (\ref{relc}) we deduce that for all $m$
$$\frac{1}{c(H)}\leq \frac{n(m)}{\log \alpha(H^{\boxtimes n(m)})} \leq \frac{n(m)}{m}.$$ 
Taking the limit of the above, we obtain the inequality 
${1/c(H)\leq \lim_{m\rightarrow \infty} n(m)/m= \scrate_m(K_2,H).}$
Next, as $\scrate_m(K_2,H)= \inf_m n(m)/m$, using the right most inequality in (\ref{relc}) we deduce that
\beqrn
\scrate_m(K_2,H) \leq \frac{n(m)}{m} &<&  \frac{n(m)}{\log \alpha(H^{\boxtimes(n(m)-1)})}
= \frac{n(m)-1}{\log \alpha(H^{\boxtimes(n(m)-1)})} \frac{n(m)}{n(m)-1}.
\eeqrn
It is clear that $\lim_{m\rightarrow \infty} n(m)=\infty$. Therefore we can conclude that the limit of the right most term in the above inequalities is equal to $1/c(H)$.
This  shows the reverse inequality
~$\scrate(K_2,H) \leq 1/c(H)$ and thus the equality ${\scrate(K_2,H) = 1/c(H)}$.
\end{proof}

{\bf Separation theorems.} Source and channel coding are often treated separately (as such, they motivate the two main branches of Shannon theory).
The main reason for this are so-called {\em separation theorems}, which roughly say that source and channel code design can be separated without asymptotic loss in the code rate in the limit of large block lengths.
Such results typically hold in a setting of asymptotically vanishing error probability~\cite{Vembu:1995}.
But when errors cannot be tolerated, Nayak, Tuncel and Rose~\cite{Nayak:2006} showed that separated codes can be highly suboptimal.
In terms of the above graph parameters, this says that in general $\scrate(G,H) \leq R(G)/c(H)$ holds (this is implied in~\cite{Nayak:2006}; see Proposition~\ref{prop:eta} for an explicit proof), but that for some families of graphs there can be a large separation: ~$\scrate(G,H) \ll R(G)/c(H)$.

\subsection{Entanglement-assisted source-channel coding}\label{subsec:defqscrate}

Here we describe the information-theoretic problems introduced above, but now in the setting where the parties are allowed to use entanglement.
We refer readers who are unfamiliar with quantum information theory to Section~\ref{sec:model}. 
We will make repeated use of the following simple lemma (already implicitly used in a similar context in~\cite{Cubitt:2010}), which we will refer to here as the {\em Orthogonality Lemma}.

\begin{lemma}[Orthogonality Lemma]\label{lem:helho}
Let~$\rho_1,\dots,\rho_\ell\in\C^{d\times d}$ be a collection of Hermitian positive semidefinite matrices. Then the following are equivalent:
\begin{enumerate}
	\item We have $\rho_i\rho_j = 0$ for every~$i\ne j \in [\ell]$.
	\item There exists a measurement consisting  of positive semidefinite matrices~$P^1,\dots,P^\ell$, $P^\perp\in\C^{d\times d}$ such that
$\Tr(P^i\rho_j) = \delta_{ij}\Tr(\rho_j)$ and  $\Tr(P^\perp \rho_j) = 0$ for every~$i,j\in[\ell]$.
\end{enumerate} 
\end{lemma}

\begin{proof}
$(1) \Rightarrow (2)$:
Let~$V_i\subseteq\C^d$ be the column space (or image) of the matrix~$\rho_i$. 
Observe that since~$\rho_i$ is Hermitian, its row space is~$V_i^*$.
Condition~$(1)$ implies that the spaces~$V_1,\dots,V_\ell$ are pairwise orthogonal.
To see this, observe that for any vectors~$u\in V_i$ and $v\in V_j$ there exist~$x,y\in \C^d$ such that~$u = \rho_i x$ and~$v=\rho_j y$. 
By Hermiticity, we have~$u^* v = x^* \rho_i\rho_j y = 0$.
Let~$P^i$ be the orthogonal projection onto~$V_i$ and let~$P^\perp = I - \sum_{i = 1}^\ell P^i$.
It is now trivial to verify that these projections satisfy the desired properties.

$(2)\Rightarrow(1)$:
Let~$V_i$ be the column space of~$\rho_i$ and let~$W_i$ be the column space of~$P^i$.
By Hermiticity, these spaces are the complex conjugates of the respective row spaces.
We claim that~$V_i\subseteq W_i$.
Indeed, expanding $\rho_i$ in its spectral decomposition, writing~$\lambda[\ell]\in\R_+$ and~$a[\ell]\in\C^d$, $\ell\in[d]$, for its eigenvalues and eigenvectors, respectively,
\beqr\label{eq:ortheq}
\sum_{\ell=1}^d\lambda[\ell] 
=
\Tr(\rho_i) = \Tr(P^i\rho_i) = 
\sum_{\ell=1}^d\lambda[\ell]\, \Tr\big(P^i a[\ell]a[\ell]^*) 
=
\sum_{\ell=1}^d\lambda[\ell]\, a[\ell]^*P^i a[\ell].
\eeqr
Since~$P^i$ has eigenvalues in~$[0,1]$, we have that~\eqref{eq:ortheq} holds if and only if~$a[\ell]^*P^i a[\ell] = 1$ for each~$\ell\in[d]$.
This implies that each eigenvector~$a[\ell]$ of~$\rho_i$ is an eigenvector of~$P^i$ with eigenvalue~$1$.
The condition~$\Tr(P^i\rho_j) = 0$ if~$i\ne j$ implies that~$W_i$ is orthogonal to~$V_j$.
Hence, for~$i\ne j$, we have ${V_i\subseteq W_i}$ and~$W_i$ is orthogonal to~$V_j$  and so the row space~$V_i$ of~$\rho_i$ is orthogonal to the column space~$V_j$ of~$\rho_j$, giving~$\rho_i\rho_j = 0$.
\end{proof}

{\bf Entanglement-assisted source coding.}
We continue by describing the special case of entanglement-assisted source coding.
The setup is as before, except now Alice and Bob have quantum registers~$\mathcal A$ and~$\mathcal B$, respectively, that are initialized to be in some entangled state.
If the source gives the parties inputs~$x\in\mathsf X$ and $u\in\mathsf U$, respectively, then their most general course of action is as follows:

\begin{enumerate}
\item After receiving her input~$x$ Alice performs a measurement on her register~$\mathcal A$ and communicates the measurement outcome to Bob;
\item After receiving both his input~$u$ and Alice's measurement outcome, Bob performs a  measurement on his register~$\mathcal B$ and obtains a measurement outcome~$y\in\mathsf X$.
\end{enumerate}

If this protocol is successful, that is, Bob gets outcome~${y=x}$ for each possible input pair~$(x,u)$, then it yields a system of matrices which we use below to derive the definition of the {\em entangled chromatic number} (Definition~\ref{def:qchrom}).
We get these matrices as follows.

Let~$\sigma$ denote the state in which the pair of registers~$(\mathcal A,\mathcal B)$ is initialized before the protocol starts.
The measurement Alice performs in step~(1) is given by a collection of positive semidefinite matrices~$A_x^1,\dots,A_x^t$ that add up to the identity.
If Alice gets outcome~$i\in[t]$, then after step~(1) Bob's register is left in a state proportional to~$\rho^i_x = \Tr_\mathcal A((A_x^i\otimes I)\sigma)$.
Note that for each~$x\in\mathsf X$ the matrices~$(\rho_x^i)_{i\in[t]}$ sum to Bob's reduced density matrix~$\rho:= \Tr_\mathcal A(\sigma)$.
Observe that Bob's input~$u$ already allows him to reduce the list of Alice's possible inputs  to a clique~$\mathcal C$ in the characteristic graph~$G$ that contains Alice's input~$x$.
After Alice does her measurement and sends the outcome~$i$, he thus knows that his register is in one of the states proportional to~$\rho_y^i$ for~$y\in \mathcal C$.
The measurement he does in step~(2) allows him to learn exactly which of these states his register is in. 
More explicitly, Bob has a measurement of positive semidefinite matrices~$(B^x)_{x\in\mathcal C}$ and~$B^\perp$ such that $\Tr(B^x\rho^i_y) = \delta_{xy}\Tr(\rho^i_y)$ and $\Tr(B^\perp\rho_x^i) = 0$ for every $x,y\in\mathcal C$.
Hence, by the Orthogonality Lemma (Lemma~\ref{lem:helho}), the states~$\rho_i^x/\Tr(\rho_x^i)$ must be pairwise orthogonal and we have
\beqr
\rho_x^i\rho_y^i =& 0 & \quad \forall i  \in [t] \text{ and every } 
 \{x,y\}\in E(G) \label{eq-intro:orth}\\
\sum_{i\in[t]}\rho_x^i &= & \rho  \quad \forall x\in V(G).\label{eq-intro:sumrho}
\eeqr
In the opposite direction we now show that such a system of matrices can be used to construct a protocol for the source-coding problem.
Cubitt et al.~\cite{Cubitt:2010} implicitly used a similar fact in the context of channel coding, albeit without proof.

\begin{proposition}\label{prop:system-to-protocol}
Suppose that there exists a collection of positive semidefinite matrices~$\rho$ and~$\{\rho_x^i\st x\in V(G),\, i\in[t]\}$ in~$\C^{d\times d}$ such that~$\Tr(\rho) = 1$ and~\eqref{eq-intro:orth} and~\eqref{eq-intro:sumrho} hold.
Then, there exists a $t$-message protocol for the source-coding problem.
\end{proposition}

The proposition follows almost directly from the following well-known theorem (see e.g., \cite[HJW Theorem, pp.~74]{Spekkens:2002}, where it is attributed to Hughston, Jozsa and Wootters).

\begin{theorem}[Hughston--Jozsa--Wootters Theorem]\label{lemma:hjw}
Let $d,t$ be positive integers, $p_1,\dots,p_t\geq 0$ satisfying $p_1 + \cdots + p_t = 1$, and let $\rho_1,\dots,\rho_t\in\C^{d\times d}$ be positive semidefinite matrices with trace~1.
Then, there exists a state $\sigma$ for a pair of registers $(\mathcal A, \mathcal B)$, and a measurement on $\mathcal A$ consisting of a collection of positive semidefinite matrices $A^1,\dots,A^t$ that add up to the identity, such that for each $i\in[t]$, we have ${\Tr_\mathcal A((A^i\otimes I)\sigma) = p_i\rho_i}$.
Moreover, $\sigma$ depends only on $p_1\rho_1 + \cdots + p_t\rho_t$. 
\end{theorem}

\begin{proof}[Proof of Proposition \ref{prop:system-to-protocol}]
It suffices to find an entangled state~$\sigma$ and measurement $\{A_x^i\st i\in[t]\}$ for each of Alice's inputs~$x$ such that~$\rho_x^i = \Tr_\mathcal A((A_x^i\otimes I)\sigma)$.
Indeed, recall that after getting his input~$u$ and Alice's measurement outcome~$i$, Bob knows that his register is in the state~$\rho_y^i/\Tr(\rho_y^i)$ for some~$y\in\mathcal C$ and clique~$\mathcal C$ in~$G$ that contains Alice's input~$x$.
By~\eqref{eq-intro:orth} and the Orthogonality Lemma there exists a measurement~$\{B_y\st y\in\mathcal C\}\cup\{B^\perp\}$ such that~$\Tr(B_y\rho_x^i) = \delta_{xy}\Tr(\rho_x^i)$, which thus allows Bob to correctly identify Alice's input~$x$.

For each~$x\in V(G)$, define the nonnegative numbers ${p^i_x = \Tr(\rho_x^i)}$, $i\in[t]$, and trace-1 matrices $\tilde\rho_x^i$ given by $\rho_x^i/p_i^x$ if $p_i^x$ is nonzero and an arbitrary trace-1 matrix otherwise.
Then, $p^1_x + \cdots + p^t_x = 1$ and $p_x^1\tilde \rho_x^1 + \cdots + p_x^t\tilde \rho_x^t = \rho$.
Hence, Theorem~\ref{lemma:hjw} gives the desired state and measurements.
\end{proof}

An optimal entanglement-assisted protocol minimizes the number of outcomes~$t$ of Alice's measurement, since these outcomes need to be communicated.
This and above considerations now justify defining the following variants of the chromatic number and Witsenhausen rate.
\medskip

\begin{definition}[Entangled chromatic number and Witsenhausen rate] \label{def:qchrom}
For a graph $G$, define $\qchrom(G)$ as the minimum integer~$t\in\N$ for which there exist $d\in\N$ and positive semidefinite matrices $\rho$ and $\{\rho_x^i\st x\in V(G),\, i\in[t]\}$ in~$\C^{d \times d}$ such that~$\Tr(\rho) = 1$ and~\eqref{eq-intro:orth} and~\eqref{eq-intro:sumrho} hold.
The {\em entangled Witsenhausen rate} is defined by
$$\qwits(G) = \lim_{m \to \infty} \frac{1}{m} \log \qchrom(G^{\boxtimes m}).$$
\end{definition}

The operational interpretation of the parameter~$\qchrom(G)$ makes it easy to see that it is sub-multiplicative with respect to strong graph products, that is~${\qchrom(G\boxtimes G') \leq \qchrom(G)\,\qchrom(G')}$.
Indeed, suppose that two sources~$\mathcal M$ and~$\mathcal M'$ admit entanglement-assisted protocols involving~$t$- and $t'$-outcome measurements done by Alice, respectively. 
For the combined source~$\mathcal M\otimes \mathcal M'$ (with characteristic graph~$G\boxtimes G'$), Alice and Bob get input {\em pairs} $(x,x')$ and $(u,u')$ respectively and Bob must learn~$(x,x')$.
By running the two protocols separately on an instance of~$\mathcal M\otimes \mathcal M'$ we get a protocol involving~$t\times t'$-outcome measurements, which gives the claim.
Alternatively, the sub-multiplicativity can be derived from Definition~\ref{def:qchrom} using simple matrix manipulations.
Sub-multiplicativity of~$\qchrom(G)$ implies that the entangled Witsenhausen rate is also given by the infimum:  ${\qwits(G) =\inf_m \log \qchrom(G^{\boxtimes m})/m}$.
\medskip

{\bf Entanglement-assisted channel coding.}
In a similar fashion as above one arrives at the following entangled variants of the independence number and Shannon capacity.

\begin{definition}[Entangled independence number and Shannon capacity]\label{def:qindep}
For a graph $H$, define $\qindep(H)$ as the maximum integer $M\in\N$ for which  there exist $d\in\N$ and positive semidefinite matrices~$\rho$ and $\{\rho_i^u\st i\in[M],\, u\in V(H)\}$ in~$\C^{d \times d}$ such that~$\Tr(\rho) = 1$ and
\beqrn
\rho^u_i\rho^v_j =& 0 & \quad \forall i\ne j \text{ and }\, \forall u,v\in V(H) \text{ s.t. }\, u=v \text{ or } \{u,v\}\in E(H),\\
\sum_{u\in V(H)} \rho^u_i &=& \rho \quad \forall i\in[M].
\eeqrn


The {\em entangled Shannon capacity} is defined by 
$$\qcap(H) = \lim_{n\to\infty}\frac{1}{n}\log\qindep(H^{\boxtimes n}).$$
\end{definition}
The parameter~$\qindep(H)$ was introduced by Cubitt et al.~\cite{Cubitt:2010}, who showed that it equals the maximum number of messages that can be sent without error using entanglement and a single use of a channel with confusability graph~$H$.
It follows that~$\qcap(H)$ equals the maximum asymptotic communication rate of such a channel when we allow for entanglement.
The parameter~$\qindep(H)$ is super-multiplicative (which also follows easily from its operational interpretation) and so in the definition of $\qcap(H)$ the limit can be replaced with the supremum.

In~\cite{Cubitt:2010} it is shown that~$\qindep(H)$ can be strictly larger than~$\alpha(H)$, meaning that the number of messages that can be sent with a single use of a channel can be increased with the use of entanglement (see also Man\v{c}inska, Severini and Scarpa~\cite{Mancinska:2012}).
This result was subsequently strengthened by Leung, Man\v{c}inska, Matthews, Ozols and Roy~\cite{Leung:2012} and Bri\"{e}t, Buhrman and Gijswijt~\cite{Briet:2012c}, who found families of graphs for which~$\qcap(H)>c(H)$.
\medskip

{\bf Entanglement-assisted source-channel coding.}
Finally, we consider the general source-channel coding problem in the entanglement-assisted setting.
As before, Alice and Bob receive inputs from a source $\mathcal M = (\mathsf X, \mathsf U, P)$ and Alice can send messages through a classical channel $\mathcal N = (\mathsf S,\mathsf V,Q)$. 
Their goal is for Bob to learn Alice's input, minimizing the number of channel uses per input sequence of a given length.
But in addition Alice and Bob possess quantum registers $\mathcal A$ and~$\mathcal B$, respectively, that are initialized to be in some entangled state~$\sigma$.
The entanglement-assisted version of an~$(m,n)$-coding scheme is as follows (Figure~\ref{fig:qwits}):

\begin{figure}
\begin{center}
 \includegraphics[width=7.5cm]{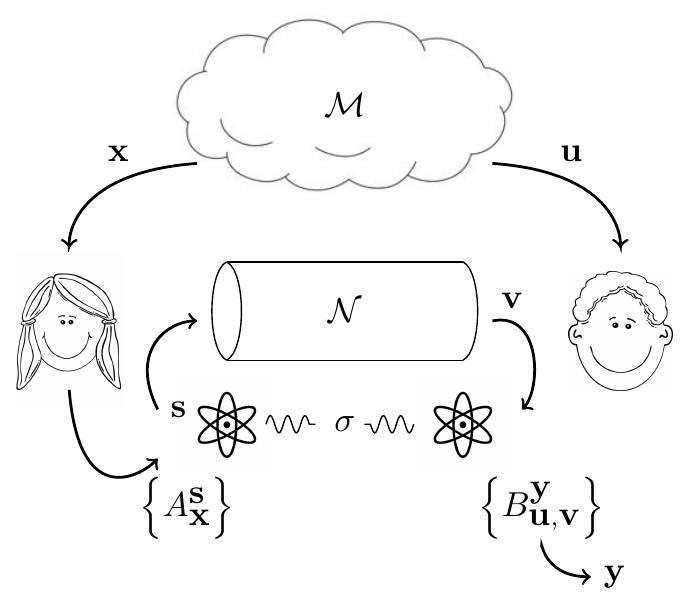}
 \caption{The  figure illustrates the entanglement-assisted source-channel coding protocol.
 After receiving a source-input ${\bf x} \in\mathsf X^m$, Alice performs a measurement $\{A_{\bf x}^{\bf s}\st {\bf s}\in \mathsf S^n\}$ on her part of an entangled state~$\sigma$ which she shares with Bob. She sends her outcome~$\bf s$ through the channel, upon which Bob---who received a source-input~$\bf u$---receives a channel-output~${\bf v} \in\mathsf V^n$.
 Bob performs a measurement $\{B_{\bf{u},\bf{v}}^{\bf y}\st {\bf y}\in \mathsf X^m\}$ on his part of~$\sigma$ and obtains an outcome~${\bf y}\in\mathsf X^m$. }
\label{fig:qwits}
\end{center}
\end{figure}

\begin{enumerate}
 \item Alice and Bob receive inputs ${\bf x}\in \mathsf X^m$ and ${\bf u} \in \mathsf U^m$, respectively, from the source $\mathcal M$;
 \item Alice performs a measurement $ \{ {A}_{\bf x}^{\bf s}\}_{{\bf s} \in \mathsf S^n} $ (which can depend on~${\bf x}$) on $\mathcal A$ and gets some sequence~$\bf s$ as outcome;
 \item Alice sends ${\bf s}$ through the channel $\mathcal N$ after which Bob receives some sequence~${\bf v} \in \mathsf V^n$;
 \item Bob performs a measurement $ \{ B_{{\bf u},{\bf v}}^{\bf y}\}_{{\bf y} \in \mathsf X^m}$ (which can depend on ${\bf u}$ and ${\bf v}$) on $\mathcal B$ and gets some sequence~${{\bf y}\in \mathsf X^m}$ as outcome.
\end{enumerate}

Using the same arguments as above one then arrives at the following variants of the cost-rate.

\begin{definition}[Entangled cost rate]\label{def:qscrate}
For graphs $G,H$ and~$m\in\N$,  define $\qscrate_m(G,H)$ as the minimum integer $n\in\N$ for which 
there exist ~$d\in\N$ and positive semidefinite matrices~$\rho$ and $\{\rho_{\bf x}^{\bf s}\st {\bf x}\in V(G^{\boxtimes m}),\, {\bf s}\in V(H^{\boxtimes n})\}$ in~$\C^{d \times d}$ such that~$\Tr(\rho) = 1$ and
\begin{align*}
\rho_{\bf x}^{\bf s}\rho_{\bf y}^{\bf t} = 0& \quad \forall{\bf x},{\bf y} \mbox{ s.t. }  \{{\bf x},{\bf y}\}\in E(G^{\boxtimes m}) \text{ and }\\ 
& \quad \forall {\bf s},{\bf t} \text{ s.t. }\,  {\bf s = \bf t} \text{ or } \{{\bf s},{\bf t}\}\in E(H^{\boxtimes n}),\\
\sum_{{\bf s}\in V(H^{\boxtimes n})} \rho_{\bf x}^{\bf s} = \rho& \quad \forall {\bf x}\in V(G^{\boxtimes m}).
\end{align*}
The {\em entangled cost rate}  is defined by
 $$\qscrate(G,H) = \lim_{m\to\infty} \frac{1}{m} \qscrate_m(G,H).$$
\end{definition}

As for the classical counterpart, we assume throughout that both graphs $G$ and $\overline{H}$ contain at least one edge, thereby excluding trivial settings.
It is not difficult to see that we regain the parameter~$\scrate(G,H)$ if we restrict the above matrices~$\rho$ and~$\rho_{\bf x}^{\bf s}$ to be~$\bset{}$-valued scalars.
Thus sharing an entangled quantum system cannot make the coding scheme worse and so $\qscrate(G,H) \leq \scrate(G,H)$.
As in the classical case, the parameter $\qscrate_m(G,H)$ is sub-additive 
(as can be easily be derived by its operational interpretation or by matrix manipulations using Definition~\ref{def:qscrate}),
hence the parameter $\qscrate(G,H)$ is well defined and can equivalently be given by~$\inf_m\qscrate_m(G,H)/m$.

We finish this subsection by noticing that the following analogue of~Lemma~\ref{lemRc} holds.

\begin{lemma}
Let $G$ and $H$ be graphs such that both $G$ and $\overline H$ have at least one edge. Then,
\beqn \qwits(G) = \qscrate(G,\overline K_2) \;\text{ and }\; 1/\qcap(H) = \qscrate(K_2, H).
\eeqn 
\end{lemma}

\begin{proof}
Notice that the graph ${\overline K_2}^{\boxtimes n}$ has $2^n$ vertices and no edges. It then follows from the definitions that
${\qscrate_m(G,\overline K_2) = \lceil \log \qchrom(G^{\boxtimes m})\rceil}$. 
The identity  ${\qwits(G)=\qscrate(G,\overline K_2)}$ follows by dividing by $m$ and letting $m$ go to infinity.

Since $K_2^{\boxtimes m} = K_{2^m}$, it follows from the definitions that $\qscrate_m(K_2,H)$ is the minimum $n\in \N$ such that ${\qindep(H^{\boxtimes n})\ge 2^m}$ or, equivalently,
$\log \qindep(H^{\boxtimes n}) \geq m$.
Now we use the same techniques as in Lemma~\ref{lemRc} to prove that ${1/\qcap(H) = \qscrate(K_2, H)}.$
\end{proof}

\subsection{Further relations to previous work}

To the best of our knowledge, neither source nor source-channel coding were considered in the context of shared entanglement before.
However, in the context of Bell inequalities,  Cameron et al.~\cite{Cameron:2007} studied the {\em quantum chromatic number}~$\chi_q(G)$, and Roberson and Man\v{c}inska~\cite{Roberson:2012} considered a variant of the quantum independence number~$\alpha_q(H)$.   
These parameters can be obtained from the respective definitions of~$\qchrom$ and~$\qindep$ given above, if we require~$\rho$ to be the identity matrix and if we further restrict the other positive semidefinite matrices  to be orthogonal projections.
Moreover, we regain~$\chi$ and~$\alpha$ if we further restrict these matrices to be $\{0,1\}$-valued scalars.
It thus follows immediately that
$$
\qchrom(G)\leq \chi_q(G) \leq \chi(G) \; \text{ and } \; \alpha(H) \leq \alpha_q(H) \leq \qindep(H).
$$

It is well-known that determining the classical chromatic and independence numbers of a graph are NP-hard problems. 
Determining the Shannon capacity and the Witsenhausen rate appears to be even harder (we do not even know if they are computable).
Despite  substantial efforts, the properties of these parameters are still only partially understood (see~\cite{Alon:2002c,Alon:2006b} and references therein). For example, the largest odd cycle for which the Shannon capacity has been determined is~$C_5$ and the decidability of the Shannon capacity and the Witsenhausen rate are still unknown.
Clearly the parameter~$\eta$ is  at least as hard to compute as~$R$ and~$c$ since it contains them as special cases.
Even less is known about the quantum variants of these parameters and determining the computational complexity of the parameters~$\qchrom, \qindep, \alpha_q, \qwits$  and~$\qcap$ is an open problem.
Very recently, however, Ji~\cite[Theorem~1 and Theorem~4]{Ji:2013} showed that it is NP-hard to decide if~$\chi_q(G) \leq 3$.

\subsection{Outline of the paper}
In Section~\ref{sec:intro} we introduced the problems and basic notions.
In Section~\ref{sec:results} we present our main results.
In Section~\ref{sec:model} we give a brief introduction to quantum information theory and we explain quantum remote state preparation. Some properties of the entangled parameters are also presented there.
The proofs of our main results are given in Sections~\ref{sec:thetabound} -~\ref{sec:scratesep}.
Finally in Section~\ref{sec:concl} we summarize our results and mention open questions.

\section{Our results}\label{sec:results}

\subsection{The entangled chromatic number and Szegedy's number}

Here we explain our lower bound on the entangled chromatic number.
We show that~$\qchrom(G)$ is lower bounded by an efficiently computable graph parameter, namely a variant of the famous {\em theta number} introduced by Szegedy~\cite{Szegedy:1994}. 
The theta number itself was originally introduced by Lov\'asz~\cite{Lovasz:1979} to solve a long-standing problem posed by Shannon~\cite{Shannon:1956}: computing the Shannon capacity of the five-cycle.
Out of the many equivalent formulations of the theta number (see~\cite{Knuth:1993} for a survey), the following is the most appropriate for our setting:
\begin{equation}
\label{opt:theta-gbar}
\begin{split}
\vartheta(G) = \min\Big\{\lambda : \, & \exists\,  Z \in \R^{V(G) \times V(G)},\;Z\succeq 0,\\
&Z(u, u) = \lambda -1 \;\;\text{for~$u \in V(G)$},\\
&Z(u, v) = -1 \;\; \text{for~$\{u,v\} \notin E(G)$}\Big\}.
\end{split}
\end{equation}
Lov\'asz~\cite{Lovasz:1979} proved that $\alpha(G)\leq \vartheta(G) \leq \chi(\overline G)$ holds (this inequality is often referred to as the Sandwich Theorem~\cite{Knuth:1993}).
The theta number is defined by a semidefinite program and it can be approximated to within arbitrary precision in polynomial time. It thus  gives a tractable and in many cases useful bound for both $\alpha$ and $\chi$.

Szegedy~\cite{Szegedy:1994} introduced the following strengthening of the theta number, which includes an extra linear constraint:
\begin{equation}
\label{opt:thplus-g}
\begin{split}
\thplus(G) = \min\Big\{\lambda : \, & \exists \, Z \in \R^{V(G) \times V(G)},\; Z\succeq 0,\\
&Z(u, u) = \lambda -1 \;\;\text{for~$u \in V(G)$},\\
&Z(u, v) = -1 \;\;\text{for~$\{u,v\} \notin E(G)$},\\
&Z(u, v) \geq -1 \;\; \text{for~$\{u,v\} \in E(G)$}\,\Big\}.
\end{split}
\end{equation}

Szegedy's number satisfies the chain of inequalities:
${\alpha(G) \leq\vartheta(G) \leq  \vartheta^+(G) \leq \chi(\overline G)}$.
Lov\'asz \cite{Lovasz:1979} proved that $\vartheta$ is {\em multiplicative} under the strong graph product, that is, $\vartheta(G\boxtimes H) = \vartheta(G)\vartheta(H)$. 
Moreover Knuth~\cite{Knuth:1993} showed that $\vartheta(\overline{G\boxtimes H}) = \vartheta(\overline G)\vartheta(\overline H)$ (recently Cubitt et al.~\cite{Cubitt:2013} showed that this identity fails for~$\vartheta^+$). 
The identities of Lov\'asz and Knuth give  for any graph $G$ and $m\in \N$:
\begin{equation}\label{relLK}
\vartheta(\overline G ^{\boxtimes m}) =\vartheta(\overline{G^{\boxtimes m}}) =\vartheta(\overline G)^m.
\end{equation}
Combining these properties of $\vartheta$ with the Sandwich Theorem shows that
\beqn
c(G) \leq \log\vartheta(G) \leq R(\overline G).
\eeqn
These inequalities capture the best known efficiently computable bounds for the Shannon capacity and the Witsenhausen rate.

Our first main result is  that the parameter $\thplus$ (and thus $\vartheta$ as well) lower bounds the entangled chromatic number and hence $\log \vartheta$ lower bounds the entangled Witsenhausen rate. For the proof we refer to Section~\ref{sec:thetabound}.

\begin{theorem}\label{thm:theta}
For any graph~$G$, we have
\begin{equation}\label{reftheta1}
 \thplus(G) \leq \qchrom(\overline{G}),
 \end{equation}
 \begin{equation}\label{reftheta2} 
\log \vartheta (G)\le R^\star (\overline G).
\end{equation}
\end{theorem}


In~\cite{Roberson:2012} it is shown  that the inequality $\vartheta(G) \leq \chi_q(\overline G)$ holds.
Theorem~\ref{thm:theta} thus strengthens this bound as it gives~${\vartheta(G) \leq \thplus(G) \leq \qchrom(\overline G)\leq \chi_q(\overline G)}$.

Beigi~\cite{Beigi:2010} and Duan, Severini and Winter~\cite{Duan:2013} proved that~$\vartheta(G)$ upper bounds~$\qindep(G)$.
The above-mentioned relations therefore imply the following sequence of inequalities:
\beqn
c(G)\leq \qcap(G)\leq \log\vartheta(G) \leq \qwits(\overline G) \leq R(\overline G).
\eeqn

Moreover, building on this work, Cubitt et al.~\cite{Cubitt:2013} proved that $\eta^*(G,H) \geq {\log \vartheta(\overline{G})} / {\log \vartheta(H)}$.

\subsection{Lower cost rates with entanglement}\label{sec:resultsrate}

We give quantitative bounds on the advantage of sharing entanglement for the following three parameters: the Witsenhausen rate, the Shannon capacity and  the cost rate of certain source-channel combinations.

\subsubsection{Quarter-orthogonality graphs and Hadamard matrices}

To show separations between the classical and entangled variants of the above-mentioned parameters, we use the following family of graphs (also considered in~\cite{Briet:2012c} for similar reasons).

\begin{definition}[Quarter-orthogonality graph $H_k$]\label{def:H_k}
For an odd positive integer~$k$, the {\em quarter-orthogonality graph}~$H_k$ has as vertex set all vectors in~$\pmset{k}$ that have an even number of~``$-1$'' entries, and as edge set the pairs with inner product~$-1$.
Equivalently, the vertices of~$H_k$ are the $k$-bit binary strings with even Hamming weight and its edges are the pairs with Hamming distance~$(k+1)/2$.
\end{definition}

The usual {\em orthogonality graph} has vertex set $\{-1,1\}^k$ (for~$k$ even) and two vertices are adjacent if they are orthogonal.
The quarter-orthogonality graph is a subgraph of the orthogonality graph, which can be seen by appending a~`$1$' to the vertices of~$H_k$.
We record the following simple lower bound on the independence number of~$H_k$ results for later use.

\begin{lemma}\label{lem:alphaHk}
For every odd positive integer $k$, we have $\alpha(H_k) \geq 2^{(k-3)/2}$.
\end{lemma}

\begin{proof}
The lemma follows by considering the subset~$W$ of all the vectors in~$V(H_k)$ (in the~$\{0,1\}^k$ setting) that have zeros in their last~$(k+1)/2$ coordinates. It is easy to see that~${|W| = 2^{(k-3)/2}}$ and that~$W$ is an independent set since it does not contain pairs of strings at Hamming distance~${(k+1)/2}$. 
\end{proof}

Some of our results  rely on the existence of certain Hadamard matrices.
A {\em Hadamard matrix} is a square matrix $A\in\pmset{\ell\times \ell}$ that satisfies $AA^{\mathsf T} = \ell I$.
The size $\ell$ of a Hadamard matrix must necessarily be 2 or a multiple of~$4$ and
the famous Hadamard conjecture (usually attributed to Paley~\cite{Paley:1933}) states that for every~$\ell$ that is a multiple of~$4$ there exists an $\ell\times \ell$ Hadamard matrix.
Although this conjecture is still open, many infinite families of Hadamard matrices are known. We will use a family constructed by~Xia and Liu~\cite{Xia:1991} (see for example~\cite{Xia:1996, Wilson:1997, Chen:1997, Xiang:1998, Xia:2006} for closely related constructions).

\begin{theorem}[Xia and Liu~\cite{Xia:1991}]\label{thm:xialiu}
Let $q$ be a prime power such that $q\equiv 1\bmod 4$.
Then, there exists a Hadamard matrix of size~$4q^2$.
\end{theorem}

We also use the following result regarding the graph $H_k$.

\begin{proposition}[Bri\"{e}t, Buhrman and Gijswijt~\cite{Briet:2012c}]\label{prop:clique}
Let $k$ be a positive integer such that there exists a Hadamard matrix of size~$k+1$.
Then,~$\omega(H_k) \geq k+1$.
\end{proposition}

\subsubsection{Improving the Witsenhausen rate}
Our first separation result  shows an exponential gap between the entangled and classical Witsenhausen rates of quarter-orthogonality graphs.

\begin{theorem}\label{thm:wits}
For every odd  integer $k$, we have
\beq\label{lem:qwitsbound}
\qwits(H_k) \leq \log(k+1).
\eeq
Moreover, if $k=4p^\ell -1$ where $p$ is an odd prime and $\ell\in \N$, then 
\beq\label{lem:witsbound}
R(H_k) \geq 0.154k-1.
\eeq
\end{theorem}

The proof of the theorem is given in Section~\ref{sec:Witsen}.

\subsubsection{Improving the Shannon capacity}
Our second separation result is a strengthening of the following result of~\cite{Briet:2012c}, which shows that for some values of~$k$, the entangled Shannon capacity of~$H_k$ can be strictly larger than its (classical) Shannon capacity.

\begin{theorem}[Bri\"{e}t, Buhrman and Gijswijt~\cite{Briet:2012c}]\label{thm:bbgmain}
Let $p$ be an odd prime 
such that there exists a Hadamard matrix of size $4p$. 
Set $k = 4p-1$.
Then,
\beqrn
\qcap(H_k) &\geq& k - 1 - 2\log(k+1),\\
c(H_k) &\leq& 0.846 k.
\eeqrn
\end{theorem}

Note that here we consider the exact bounds on $\qcap(H_k)$ and $c(H_k)$ rather than the asymptotic ones as originally written in~\cite{Briet:2012c}.
It is not known if Hadamard matrices of size~$4p$ exist for infinitely many primes~$p$.
Theorem~\ref{thm:bbgmain} requires the existence of Hadamard matrices due to the technique used to lower bound~$\qcap(H_k)$, which originates from~\cite{Leung:2012}.
It also requires that~$k$ is of the form~$rp - 1$ for some odd prime~$p$ and positive integer~$r\geq 4$ due to the technique used to upper-bound~$c(H_k)$, which is based on a result of Frankl and Wilson~\cite{Frankl:1981}.

Here we relax the conditions in Theorem~\ref{thm:bbgmain} and our result does not rely anymore on the existence of a Hadamard matrix. We show the existence of an infinite family of quarter-orthogonality graphs whose entangled capacity exceeds their Shannon capacity. 
\begin{theorem}\label{thm:cap}
For every odd integer~$k \geq 11$, we have 
\begin{equation}\label{lem:qcapbound}
\qcap(H_k) \geq (k-1)\left(1-\frac{2\log(k+1)}{k-3}\right).
\end{equation}
Moreover, if $k=4p^\ell -1$ where $p$ is an odd prime and $\ell\in \N$, then 
\begin{equation}\label{lem:newshan-cap}
c(H_k) \leq 0.846\, k.
\end{equation}
\end{theorem}

We prove Theorem~\ref{thm:cap} in  Section~\ref{sec:Shannon}.
To prove (\ref{lem:qcapbound}) we use a technique that is based on quantum remote state preparation~\cite{Bennett:2001}.
This proof technique appears not to have been considered before in the context of  zero-error entanglement-assisted communication.
The proof of (\ref{lem:newshan-cap})  combines an instance of the linear algebra method due to Alon~\cite{Alon:1998} with a construction of certain low-degree polynomials over a finite field for a low-degree representations of the OR-function due to Barrington, Beigel and Rudich~\cite{Barrington:1994}.
Roughly this combination was previously used in the context of Ramsey graphs~\cite{Gopalan:2006}.

\subsubsection{Improving on source-channel codes}\label{sec:improvingscrate}
Our last contribution concerns the combined source-channel problem for a source and channel that both have $H_k$ as characteristic and confusability graph, respectively. The result is the following.

\begin{theorem}\label{thm:scratesep}
Let~$p$ be an odd prime and $\ell\in\N$ such that there exists a Hadamard matrix of size~$4p^{\ell}$.
Set $k = 4p^\ell - 1$.
Then, 
\beq\label{eq:qscratebound}
\qscrate(H_k,H_k) \leq \frac{\log (k+1)}{(k-1)\left(1-\frac{2\log(k+1)}{k-3}\right)
},
\eeq
\beq\label{eq:scratebound}
\scrate(H_k,H_k) >  \frac{0.154\, k-1} {k-1-\log(k+1)}.
\eeq
\end{theorem}

The proof of Theorem~\ref{thm:scratesep} is given in Section~\ref{sec:scratesep}.
The bound on the entangled source-channel cost rate is obtained by concatenating an entanglement-assisted coding scheme for a source with one for a channel. In this way, one obtains a ``separated'' coding scheme for the source-channel problem, see Section~\ref{sec:sepschemes} for details. There we show that the asymptotic cost rate of a separate coding scheme is $\qwits(H_k)/\qcap(H_k)$ and thus  $\qscrate(H_k,H_k) \leq \qwits(H_k)/\qcap(H_k)$. 
The bound for the classical parameter $\eta(H_k,H_k)$ relies on the No-Homomorphism Lemma due to Albertson and Collins~\cite{Albertson:1985} and the fact that $H_k$ is vertex-transitive.
Let us point out that Theorem \ref{thm:scratesep} holds for an infinite family of graphs.
This follows  from the result of Xia and Lu \cite{Xia:1991}  in Theorem~\ref{thm:xialiu},
since  there exist infinitely many $(p,\ell)$-pairs such that $p^{\ell/2}\equiv 1\bmod 4$. (For instance, for   $p=5$ and $\ell=2i$ with $i\in \N$, 
$5^i=(4+1)^i \equiv 1 \bmod 4$.)

Hence, for any $k$ satisfying the condition of the theorem, we have an exponential separation between the entangled and the classical source-channel cost rate as
\beqn
\qscrate(H_k,  H_k )  \leq  \frac{\qwits(H_k)}{\qcap(H_k)}
 \leq  O\Big(\frac{\log k}{k}\Big) \quad
 \text{while}  \quad  \scrate(H_k,H_k)  \geq \Omega (1).
\eeqn
As shown in~\cite{Nayak:2006}, a large separation~$\scrate(G,H) \ll R(G)/c(H)$ exists for some graphs.
But this is not the case for our source-channel combination using $G=H=H_k$.
Indeed, $\scrate(H_k,H_k) \geq \Omega(1)$ and 
\beqn
\frac{R(H_k)}{c(H_k)} \leq  \frac{\log \chi(H_k)}{\log\alpha(H_k)}  
 \leq  \frac{2(k-1)}{k-3} \leq O(1),
\eeqn
 where in the second last inequality we use that ${\log \chi(H_k) \leq \log |V(H_k)| = k-1}$ and, from Lemma~\ref{lem:alphaHk}, that~${\log \alpha(H_k) \geq (k-3)/2}$.

\section{Entanglement-assisted source-channel coding}\label{sec:model}

In this section we describe quantum remote state preparation and show how to obtain a coding scheme for the source-channel problem from concatenating a coding scheme for the source with one from the channel. We begin by describing the elements of quantum information theory that appear in the subsequent discussion.
For more on quantum information theory we refer to Nielsen and Chuang~\cite{Nielsen:2000}.

\subsection{Quantum states and operations}

Recall that a {\em positive semidefinite matrix}  is a Hermitian matrix whose eigenvalues are  non-negative.
A {\em quantum register} is an abstract physical system with which experimenters (Alice and Bob) may interact.
A quantum register is represented by a finite-dimensional complex vector space and the register is $d$-dimensional if this vector space is~$\C^d$.
The set of possible~{\em states} of a $d$-dimensional quantum register is formed by the $d\times d$ complex positive semidefinite matrices whose trace equals~$1$.
When such a state is $\rho$, the quantum register~$\mathcal A$ is said to be {\em in} state~$\rho$.
A state with rank 1 is called a {\em pure state}.
An experimenter can alter a state~$\rho\in\C^{d\times d}$ of his or her register by performing a unitary transformation. This is a mapping $\rho\mapsto U\rho  U^*$, where $U\in\C^{d\times d}$ satisfies~$UU^* = I$.

The possible states of a {\em pair} of quantum registers~$(\mathcal A,\mathcal B)$ are the trace-$1$ positive semidefinite matrices in~${\C^{d_A\times d_A}\otimes\C^{d_B\times d_B}}$. Here,  $d_A$ and $d_B$ are the dimensions of~$\mathcal A$ and $\mathcal B$, respectively, and 
$\C^{d_A\times d_A}\otimes\C^{d_B\times d_B}$ is  the tensor product space, consisting of all linear combinations of matrices of the form $\rho_A\otimes \rho_B$, where $\rho_A\in\C^{d_A\times d_A}$ and $\rho_B\in\C^{d_B\times d_B}$.
The  pair of systems $(\mathcal A,\mathcal B)$ is said to be {\em entangled} if it is in a state $\rho$ which is {\em not a convex combination of states of the form} $\rho_A\otimes \rho_B$.
If Alice performs a unitary~$U\in\C^{d_A\times d_A}$ on her register~$\mathcal X$ then the state~$\rho$ of the system~$(\mathcal A,\mathcal B)$ is mapped to~$(U\otimes I)\rho (U\otimes I)^*$, and similarly if Bob performs a unitary on his register.

A {\em $t$-outcome measurement} is a collection ${\mathsf M = \{M_i\in\C^{d\times d}\st i\in[t]\}}$ of positive semidefinite matrices~$M_i$ that satisfy $\sum_{i=1}^tM_i = I$, where $I$ is the identity matrix.
A measurement describes an experiment which one may perform on a $d$-dimensional quantum register.
If Alice performs a $t$-outcome measurement~$\mathsf M$ on a register~$\mathcal A$ which is in a state~$\rho$, then she will observe a random variable~$\lambda$ over the set~$[t]$ whose probability distribution is given by $\Pr[\lambda = i] = \Tr(M_i\rho)$.
In the event that $\lambda =i$, we say that Alice gets measurement outcome~$i$.

Below we will consider settings where Alice and Bob hold (possibly entangled) quantum registers~$\mathcal A$ and~$\mathcal B$, respectively, and they each perform a measurement.
For this we introduce a linear operator called the {\em partial trace}.
For matrices~${A\in\C^{d_A\times d_A}}$ and $B\in\C^{d_B\times d_B}$ define~${\Tr_{\mathcal A}(A\otimes B) = \Tr(A)B}$ and $\Tr_{\mathcal B}(A\otimes B) = A\Tr(B)$, and extend these definitions in a linear fashion to all matrices of~$\C^{d_A\times d_A}\otimes\C^{d_B\times d_B}$.

Suppose that the pair~$(\mathcal A,\mathcal B)$ is in the state~$\rho$ and that Alice performs a $t$-outcome measurement~$\mathsf M$ on~$\mathcal A$.
Then, the probability that Alice gets measurement outcome~$i$ equals $p_i=\Tr\big((M_i\otimes I)\rho\big)$.
Moreover, in the event that Alice gets measurement outcome~$i$, Bob's register~$\mathcal B$ is left in the state $\rho^i = \Tr_{\mathcal A}\big((M_i\otimes I)\rho\big)/p_i$.
If Bob now performs an~$r$-outcome measurement $\mathsf M'$ on~$\mathcal B$, then the probability that he gets outcome~$j\in[r]$ equals~$\Tr_{\mathcal B}(M'_j\rho^i)$.

\subsection{Remote state preparation}\label{sec:rsp}

Suppose that Alice has in mind a pure state~$\rho = uu^*$ where $u\in\C^{d}$ is some unit vector that is unknown to Bob.
{\em Remote state preparation}~\cite{Bennett:2001} is a protocol that enables the parties to prepare a quantum register belonging to Bob in the state~$\rho$ using only local measurements on a pair of entangled quantum registers and classical communication from Alice to Bob.
This task can be achieved using the teleportation scheme of Bennett et al.~\cite{Bennett:1993}, which allows to remotely prepare a $d$-dimensional state $\rho$ with the communication of one among $d^2$ distinct messages. (We refer interested readers to~\cite{Bennett:1993} and~\cite[pp.~26--28]{Nielsen:2000} for the details of the teleportation scheme.) However, for certain states, remote state preparation can be performed with less communication. 
In particular, for the case where the vector~$u$ has only entries of absolute value~$d^{-1/2}$ we now describe a remote state preparation protocol that requires only sending one among $d$ distinct messages, or equivalently $\ceil{\log d}$ bits.

To remotely prepare the  state ~$\rho = uu^*$, where $u\in\C^{d}$ is a vector with entries of modulus~$d^{-1/2}$ that is unknown to Bob, in a $d$-dimensional register~$\mathcal Y$ that he possesses the two parties can use the following protocol:
Assume Alice has a $d$-dimensional quantum register~$\mathcal X$, such that $(\mathcal X,\mathcal Y)$ is in the {\em maximally entangled state}~$\sigma = vv^*$ defined by the vector
\beqn
v = \frac{1}{\sqrt{d}}\sum_{\ell=1}^de_\ell\otimes e_\ell,
\eeqn
where~$e_\ell$ denotes the~$\ell$th canonical basis vector.

\begin{enumerate}
\item Alice performs on her register~$\mathcal X$ the unitary transformation given by~$U = \sqrt{d}\Diag(u)$, where~$\Diag(u)$ is the diagonal matrix whose main diagonal is  vector~$u$.
\item Next, Alice performs on her register the discrete Fourier transform~$F\in\C^{d\times d}$, which is given by the unitary matrix~$F_{\ell, m} = e^{2\pi i (\ell-1)(m-1)/d}/\sqrt{d}$.
\item She then measures in the canonical basis (\ie, she performs the measurement $\{e_\ell e_\ell^*\st \ell\in[d]\}$ on her register) and gets an outcome~$\ell\in[d]$ which she communicates to Bob.
\item Last, Bob performs on his register~$\mathcal Y$ the unitary given by the diagonal matrix whose main diagonal is the vector $(e^{-2\pi i (\ell-1)(m-1)/d})_{m=1}^d$.
\end{enumerate}

The correctness of the protocol follows easily from the following observations.
After step (2) the state of the register pair~$(\mathcal X,\mathcal Y)$ is given by $(FU\otimes I)vv^*(FU\otimes I)^*$.
Notice that if $a_\ell,b_\ell$ are the $\ell$th column and row of $FU$, resp., then
\beqn
(FU\otimes I)v
=
\frac{1}{\sqrt{d}}\sum_{\ell=1}^d a_\ell\otimes e_\ell
=
\frac{1}{\sqrt{d}}\sum_{\ell=1}^d e_\ell\otimes b_\ell
\eeqn 
Since the~$\ell$th  row of the matrix $FU$ is given by
\beqn
\sum_{m = 1}^d \gamma_m\,  e_m \text{ where } \gamma_m = u_m\, e^{2\pi i (\ell-1)(m-1)/d},
\eeqn
this is thus exactly the state of Bob's register~$\mathcal Y$ after step~(3).
Thus, after step~(4) Bob's register is in the state~$\rho= uu^*$ as desired.

\subsection{Separate coding schemes}\label{sec:sepschemes}

By concatenating a coding scheme for a source with one for a channel, one obtains a coding scheme for the combined source-channel problem.
For this to work, the number of bits one can send perfectly with~$n$ uses of the channel must be at least as large as the number of bits required to solve~$m$ instances of the source problem.
In other words, for a source with characteristic graph~$G$ and a channel with confusability graph~$H$, we need  the condition $\chi(G^{\boxtimes m}) \leq \alpha(H^{\boxtimes n})$  in order to send length $m$ source-input sequences with~$n$ uses of the channel.
If this condition holds, then it follows that $\scrate_m(G,H) \le n$.
The same type of reasoning holds also in the entanglement-assisted case. We have just proved the following simple lemma, which can alternatively be shown using the definition of a graph homomorphism and simple matrix manipulations.

\begin{lemma}\label{lem:qsep}
Given graphs   $G,H$ and positive integers $n,m$, we have
\begin{equation}\label{eqsepc1}
\chi(G^{\boxtimes m}) \le \alpha(H^{\boxtimes n}) \Longrightarrow \scrate_m(G,H)\le n, 
\end{equation}
\begin{equation}\label{eqsepq1}
\qchrom(G^{\boxtimes m}) \leq \qindep(H^{\boxtimes n}) \Longrightarrow \scrate^\star_m(G,H)\le n.
\end{equation}
\end{lemma}

We now relate the minimum cost rate to the ratio of the Witsenhausen rate  and the Shannon capacity in both classical and entanglement-assisted cases.

\begin{proposition}\label{prop:eta}
Let $G$ and $H$ be graphs and assume that both $G$ and $\overline H$ have at least one edge. Then,
\begin{equation}\label{eqsepc2}
\eta(G,H) \le \frac{R(G)}{c(H)}= \lim_{m\rightarrow \infty} \frac{1}{m} \epsilon_m(G,H),
\end{equation}
\begin{equation}\label{eqsepq2}
\qscrate(G,H) \leq \frac{\qwits(G)}{\qcap(H)} =  \lim_{m\rightarrow \infty} \frac{1}{m} \epsilon_m^\star(G,H),
\end{equation}
where $\epsilon_m(G,H):= \min\{n: \chi(G^{\boxtimes m})\le \alpha(H^{\boxtimes n})\}$ and $\epsilon_m^\star(G,H):= \min\{n: \chi^\star(G^{\boxtimes m})\le \alpha^\star(H^{\boxtimes n})\}$.
\end{proposition}

\begin{proof}
We  show (\ref{eqsepc2}); we omit the proof of (\ref{eqsepq2}) which is analogous (and uses (\ref{eqsepq1})).
 From (\ref{eqsepc1}) we have the inequality:
  $\scrate_m(G,H)\le \epsilon_m(G,H),$ 
  which implies 
${\eta(G,H)\le \lim_{m\rightarrow \infty} \epsilon_m(G,H)/m}$. Next we show that this limit is equal to $R(G)/c(H)$, which concludes the proof of (\ref{eqsepc2}).
Setting $n=\epsilon_m(G,H)$,  we have that ${\alpha(H^{\boxtimes (n-1)})< \chi(G^{\boxtimes m})\le \alpha(H^{\boxtimes n})}$, implying
\beqrn
\frac{R(G)}{c(H)} & \le & {\log \chi(G^{\boxtimes m})\over m} \frac{n}{ \log \alpha(H^{\boxtimes n})} 
\le
\frac{n}{m}
< \frac{n}{n-1} \frac{\log \chi(G^{\boxtimes m})}{m} \frac{n-1}{\log\alpha(H^{\boxtimes n-1})}.
\eeqrn
Taking limits as $m\rightarrow \infty$ in the right most terms we obtain that $R(G)/c(H)$ is equal to $\lim_{m\rightarrow \infty} \epsilon_m(G,H)/m$.
\end{proof}

\section{Lower bound on the entangled chromatic number}\label{sec:thetabound}

In this section we  prove Theorem \ref{thm:theta}, which we restate here for convenience.

\begin{reptheorem}{thm:theta}
For any graph~$G$, we have
\begin{equation}\label{reftheta1-2}
 \thplus(G) \leq \qchrom(\overline{G}),
 \end{equation}
 \begin{equation}\label{reftheta2-2} 
\log \vartheta (G)\le R^\star (\overline G).
\end{equation}
\end{reptheorem}

We will use the  following  simple fact about  positive semidefinite matrices with a special block form (which can be found, e.g., in \cite[Lemma 2.8]{Gvozdenovic:2008}).
\begin{lemma}
\label{lem:block}
Let~$X$ be a~$t \times t$ block matrix, with a  matrix~$A$ as diagonal blocks and a matrix $B$ as non-diagonal blocks, of the form
\beqn\label{eq:block}
X = 
\underbrace{\begin{pmatrix}
A & B & \ldots & B \\
B & A & \ldots & B \\
\vdots & \vdots & \ddots & \vdots \\
B & B  & \ldots & A
\end{pmatrix}}_\text{t blocks}.
\eeqn
Then, $X \succeq 0$ if and only if $A - B \succeq 0$ and $A + (t-1)B \succeq 0$.
\end{lemma}

\begin{proof}[Proof of Theorem~\ref{thm:theta}]
We show that relations (\ref{reftheta1-2}) and (\ref{reftheta2-2}) hold for the graph $\overline G$.
First we observe that (\ref{reftheta2-2}) follows easily from (\ref{reftheta1-2}).
Indeed, (\ref{reftheta1-2}) combined with the identity (\ref{relLK}) implies
$\vartheta(\overline G)^m= \vartheta(\overline{G^{\boxtimes m}}) \le \chi^\star(G^{\boxtimes m})$ and thus 
${\log \vartheta (\overline G)\le R^\star (G)}$ follows after taking limits.

We now prove (\ref{reftheta1-2})  for the graph $\overline G$, i.e., we show that 
 the inequality $\vartheta^+(\overline G)\le \chi^\star(G)$ holds.
For this let $\rho,{\{\rho_u^i \st u \in V(G), \,i \in [t]\}}$ be a set of positive semidefinite matrices which form a solution for $\qchrom(G)=t$. We may assume that 
$\langle \rho,\rho\rangle =1$.
Define the matrix $X$, indexed by all pairs $\{u,i\}\in V(G)\times [t]$, with entries $X_{ui,vj} := \langle \rho_u^i,\rho_v^j \rangle$. By construction, $X$ is a non-negative positive semidefinite matrix which satisfies $X_{ui,vi} = 0$ for every $\{u,v\} \in E(G)$ and $i \in [t]$.

For any element $\sigma$ of Sym($t$), the group of permutations of $[t]$, we define the new (permuted)  matrix ${\sigma(X) = (X_{u \sigma(i),v \sigma(j)})}$.
Then we average the matrix $X$ over the group Sym($t$), obtaining the new matrix
$$Y = \frac{1}{\lvert \text{Sym}(t) \lvert} \sum_{\sigma \in \text{Sym}(t)} \sigma(X).$$
By construction, the matrix $Y$ is invariant under any permutation of $[t]$, \ie, $\sigma(Y)=Y$ for any $\sigma\in$ Sym($t$).
  Therefore, $Y$ has the block form of Lemma~\ref{lem:block} with, moreover, 
  \begin{equation}\label{eqAuv}
  A_{uv} = 0 \ \text{ for all } \{u,v\} \in E(G).
  \end{equation}
As each matrix $\sigma(X)$ is positive semidefinite, the matrix $Y$ is positive semidefinite as well.
From Lemma~\ref{lem:block}, this implies that $A-B$ and $A +(t-1)B$ are positive semidefinite matrices.
Using the definition of the matrix $X$ combined with the properties of the matrices $\rho^i_u$  and  the invariance of $Y$, we obtain the following relation for  any $u,v \in V(G)$:
\beqrn
1 
=
\langle \rho, \rho \rangle = \langle \sum_{i\in [t]}\rho_u^i,\sum_{j\in [t]}\rho^j_v\rangle = 
 \sum_{i \in [t]} \sum_{j \in [t]} \langle \rho_u^i, \rho_v^j \rangle 
 = \sum_{i \in [t]} \sum_{j \in [t]} X_{ui,vj} = \sum_{i \in [t]} \sum_{j \in [t]} Y_{ui,vj},
\eeqrn
implying that $ t(A_{uv} + (t-1)B_{uv}) = 1$ holds since
\beq \label{lowth}
\sum_{i \in [t]} \sum_{j \in [t]} Y_{ui,vj}= t \sum_{j \in [t]} Y_{ui,vj} 
 = t(A_{uv} + (t-1)B_{uv}).
\eeq
We are now ready to define a matrix $Z$ which is a feasible solution for the program (\ref{opt:thplus-g}) defining $\thplus(G)$.
Namely, set ${Z = t(t-1)(A-B)}$. Then, $Z$ is a positive semidefinite matrix.
For any edge $\{u,v\}\in E(G)$, the relations (\ref{eqAuv}) and (\ref{lowth}) give  $A_{uv}=0$ and $t(t-1)B_{uv}=1$ and thus $Z_{uv}=-1$.
For a non-edge $\{u,v\}$, relation (\ref{lowth}) combined with the fact that $A_{uv}\ge 0$ implies that $Z_{uv}\ge -1$.
Finally, for any $u\in V(G)$, relation (\ref{lowth}) combined with the fact that $B_{uu}\ge 0$ implies that $Z_{uu}\le t-1$.
 Define the vector $c$ with entries $c_u = t-1-Z_{uu}\ge 0$ for $u\in V(G)$, the diagonal matrix $D(c)$ with $c$ as diagonal,  and the matrix $Z'=Z + D(c)$. Then, $Z'$  is positive semidefinite and satisfies all the conditions of  the program ~(\ref{opt:thplus-g}) defining $\vartheta^+(\overline G)$. 
 This shows that  $\thplus(\overline{G}) \leq \qchrom(G)$, which concludes the proof.
\end{proof}

\section{Separation between classical and entangled Witsenhausen rate}\label{sec:Witsen}

Here we prove Theorem~\ref{thm:wits}, which shows an exponential separation between the classical and entangled-assisted Witsenhausen rate for the family of graphs $H_k$ in Definition~\ref{def:H_k}.
We repeat the formal statement of the theorem here for convenience.

\begin{reptheorem}{thm:wits}
For every odd  integer $k$, we have
\beq\label{lem:qwitsbound-2}
\qwits(H_k) \leq \log(k+1).
\eeq
Moreover, if $k=4p^\ell -1$ where $p$ is an odd prime and $\ell\in \N$, then 
\beq\label{lem:witsbound-2}
R(H_k) \geq 0.154k-1.
\eeq
\end{reptheorem}

\subsection{Upper bound on the entangled Witsenhausen rate}

Here we prove the upper bound~\eqref{lem:qwitsbound-2} on~$\qwits(H_k)$ stated in Theorem~\ref{thm:wits}.
We first prove that $\qchrom(G) \leq \xi'(G)$.
This inequality can be derived from the fact that $\qchrom(G)\leq \chi_q(G)$ and a result of~\cite{Cameron:2007} stating that~$\chi_q(G)\leq \xi'(G)$.
We give self-contained proofs of the implied bound on~$\qchrom(G)$ here.

\begin{lemma}\label{lem:xibound}
For every graph~$G$, we have~$\qchrom(G) \leq \xi'(G)$. 
\end{lemma}

We give two simple proofs of this fact, the first operational in spirit, the second  combinatorial.
Although the combinatorial proof is more elementary, the operational proof will prepare the reader for the slightly more involved proof of Lemma~\ref{lem:rsp-channel} below, for which an operational proof appears to be most natural.
In addition, it follows almost immediately from the operational proof that~$\qchrom(G)\leq \xi(G)^2$; the only thing to change in the proof is to replace the remote state preparation scheme described in Section~\ref{sec:rsp} by quantum teleportation.

\begin{proof}
Operational proof:
Consider a source ${\mathcal M = (\mathsf X, \mathsf U, P)}$ with characteristic graph $G$. Let $d =\xi'(G)$ and let~$f$ be a $d$-dimensional orthogonal representation of~$G$ such that all the entries of the vectors have modulus one.
By the discussion in Section~\ref{subsec:defqscrate} it suffices to find an entanglement-assisted protocol for~$\mathcal M$ that involves only~$d$-outcome measurements on Alice's part.
To this end, let us recall the observation that Bob's input~$u$ already allows him to reduce the list of Alice's possible inputs to a clique~$\mathcal C$ in~$G$ that contains Alice's actual input~$x$.
Next, notice that the set of states~$f(y)f(y)^*/d$ for~$y\in \mathcal C$ are pairwise orthogonal.
This suggests the following protocol.
First the parties perform the remote state preparation protocol (Section~\ref{sec:rsp}) to put a quantum register belonging to Bob in the state~$f(x)f(x)^*/d$.
Now Bob performs the measurement with outcomes in~$\mathcal C\cup\{\perp\}$ as promised to exist by the Orthogonality Lemma (Lemma~\ref{lem:helho}) to learn which of the states~$f(y)f(y)^*/d$ his register is in, thus learning~$x$.
The result now follows because the remote state preparation involves only~$d$-outcome measurements.

Combinatorial proof:
Set~$d = \xi'(G)$, $\omega_d = e^{2i\pi/d}$ and let $h_i = [1, \omega_d^{i},\omega^{2i}\dots,\omega_d^{(d-1)i}]^{\mathsf T}\in\C^d$ for each $i\in[d]$.
It is not hard to see that $\{h_1,h_2,\dots,h_{d}\}$ is a complete orthogonal basis for~$\C^d$.
Set $\rho = I/d$.
Then~$\Tr(\rho) = 1$.
Let $f:V(G) \to \C^d$ be an orthonormal representation of~$G$ where each vector~$f(u)$ is such that $\overline{f(u)_i}f(u)_i = 1/d$ for every~$i\in[d]$, as guaranteed to exist by the fact that~$\xi'(G) = d$.
For every $u\in V(G)$ and $i\in[d]$, define $\rho_u^i = \big(f(u)\circ h_i\big)\big(f(u)\circ h_i\big)^*$, where $\circ$ denotes the entrywise product.
Then,
\beqn
\langle f(u)\circ h_i,f(v)\circ h_j\rangle = 
\condeq{
	\langle h_i,h_j\rangle/d & \text{if $u = v$},\\[.1cm]
	\langle f(u),f(v)\rangle & \text{if $i=j$}.
}
\eeqn
It follows that 
 for every vertex~$u\in V(G)$, we have 
 ${\sum_{i \in [d]} \rho_u^i = I/d = \rho}$.
Moreover, for each~${\{u,v\}\in E(G)}$ and $i\in [d]$, we have $\rho_u^i\rho_v^i = 0$.
As the matrices $\rho, \rho_u^i$ are also positive semidefinite, they satisfy all the requirements of Definition~\ref{def:qchrom} and so~$\qchrom(G) \leq d$.
\end{proof}

The above lemma gives a bound on the entangled chromatic number of powers of~$H_k$ from which it will be easy to get the upper bound on~$R(H_k)$ given in~\eqref{lem:qwitsbound-2}.

\begin{lemma}\label{lem:qchrom}
Let $k$ be an odd positive integer and $m\in\N$. Then,
$
\qchrom(H_k^{\boxtimes m}) \leq (k+1)^{m}.
$
Moreover, if there exists a Hadamard matrix of size $k+1$, then equality holds.
\end{lemma}

\begin{proof}
Consider the map~$f:V(H_k)\to\pmset{k+1}$ defined by~$f(u) = (u,1)^{\mathsf T}$ (\ie, $f$ sends a vertex $u\in \pmset{k}$ to itself and appends a ``1'' to it).
Then clearly~$f$ is an orthogonal representation of~$H_k$ and since its vectors have entries of modulus one we get~${\xi'(H_k)\leq k+1}$.
By Lemma~\ref{lem:xibound} and sub-multiplicativity we therefore have~${\qchrom(H_k^{\boxtimes m}) \leq \qchrom(H_k)^m \leq (k+1)^{m}}$. 

We now prove that if there exists a Hadamard matrix of size $k+1$ then also the reverse inequality holds: ${\qchrom(H_k^{\boxtimes m}) \geq (k+1)^{m}}$. Recall from Proposition~\ref{prop:clique} that the existence of a Hadamard matrix of size~$k+1$ implies~${\omega(H_k) \geq k+1}$.
Combining this with Theorem \ref{thm:theta} and the Sandwich Theorem gives that for every positive integer~$m$, we have
\beqn
\qchrom(H_k^{\boxtimes m}) \geq \vartheta(\overline{H_k^{\boxtimes m}}) \geq \omega(H_k^{\boxtimes m}) \geq \omega(H_k)^m \geq (k+1)^m,
\eeqn
where the third inequality uses the simple fact that if a subset~$W\subseteq V(G)$ forms a clique in a graph~$G$, then the set~$W^m$ of $m$-tuples of elements from~$W$ forms a clique in~$G^{\boxtimes m}$.
\end{proof}

The bound~\eqref{lem:qwitsbound-2} now follows as a simple corollary.

\begin{corollary}
Let~$k$ be an odd positive integer. Then~$\qwits(H_k) \leq \log(k+1)$.
\end{corollary}

\begin{proof}
By Lemma~\ref{lem:qchrom} we have $${\qwits(H_k) = \inf_m \log \qchrom(H_k^{\boxtimes m})/m \leq \log\qchrom(H_k) \leq \log(k+1)}.$$
\end{proof}

We also record the following additional corollary, which we use later in Section~\ref{sec:scratesep}.

\begin{corollary}\label{cor:clique}
For every odd integer $k$ such that there is a Hadamard matrix of size $k+1$, we have $\omega(H_k^{\boxtimes m}) = (k+1)^m$.
\end{corollary}

\begin{proof}
Combining Proposition~\ref{prop:clique} and Lemma~\ref{lem:qchrom} gives the result.
\end{proof}

\subsection{Lower bound on the classical Witsenhausen rate}

To prove the lower bound~\eqref{lem:witsbound-2} on~$R(H_k)$ stated in Theorem~\ref{thm:wits} we use the following upper bound on the classical independence number of the graphs $H_k^{\boxtimes m}$ for certain values of $k$. 

\begin{lemma}\label{lem:newshan}
Let $p$ be an odd prime number, $\ell\in\N$ and set $k= 4p^\ell-1$.
Then, for every $m\in\N$, we have
\begin{equation}\label{eq:newshan}
\alpha(H_k^{\boxtimes m})  \leq  \left( {k\choose 0}  + {k\choose 1} + \cdots + {k\choose p^\ell-1}\right)^m 
\leq  2^{k\, m\, H(3/11)}
 < 2^{0.846 \, k\,m},
 \end{equation}
%
 where $H(t)=  -t\log t - (1-t)\log(1-t)$ is the binary entropy function.
\end{lemma}

The proof of this lemma is an instance of the linear algebra method due to Alon~\cite{Alon:1998} (see also Gopalan~\cite{Gopalan:2006}), which we recall below for completeness.
Let $G$ be a graph and $\F$ be a field. 
Let $\mathcal F\subseteq \F[x_1,\dots,x_k]$ be a subspace of the space of $k$-variate polynomials over~$\F$.
A {\em representation} of $G$ over~$\mathcal F$ is an assignment $\big((f_u,c_u)\big)_{u\in V(G)}\subseteq \mathcal F\times \F^k$ of polynomial-point pairs to the vertices of $G$ such that
\begin{equation*}
 f_u(c_u)  \neq 0 \quad \forall u\in V(G) \; \text{ and } 
 f_u(c_v)  =0 \quad \forall u\ne v\in V(G) \text{ with } \{u,v\}\not\in E(G).
\end{equation*}
%
%

\begin{lemma}[Alon~\cite{Alon:1998}]\label{lem:polybound}
Let $G$ be a graph, $\F$ be a field, $k\in \N$ and $\mathcal F$ be a subspace of $\F[x_1,\dots,x_k]$. If ${\big((f_u,c_u)\big)_{u\in V}\subseteq\mathcal F\times\F^k}$ represents $G$, then  ${\alpha(G^{\boxtimes n})\leq \dim(\mathcal F)^n}$ for all $n\in\N$.
\end{lemma}

We get a representation for the graph $H_k$, for  ${k = 4p^\ell-1}$, from the following result of Barrington, Beigel and Rudich~\cite{Barrington:1994}
(see~\cite[Lemma 5.6]{Yekhanin:2012} for the statement as it appears below). 

\begin{lemma}[Barrington, Beigel and Rudich~\cite{Barrington:1994}]\label{lem:bbr}
Let $p$ be a prime number and let $k$, $\ell$ and $w$ be integers such that $k>p^\ell$.
There exists a multilinear polynomial $f\in\Z_p[x_1,\dots,x_k]$ of degree $\deg(f) \leq p^\ell-1$ such that for every $c\in \bset{k}$, we have
\beqn
f(c)\equiv \left\{
\begin{array}{ll}
1 & \text{if } c_1 + c_2 + \cdots c_k\equiv w\bmod p^\ell\\
0 & \text{otherwise.}
\end{array}
\right.
\eeqn
\end{lemma}

With this we can now prove Lemma~\ref{lem:newshan}.

\begin{proof}[Proof of Lemma~\ref{lem:newshan}]
Let $c\in\bset{k}$ be a string such that its Hamming weight $|c|$ is even and satisfies $|c|\equiv 0\bmod p^\ell$.
Then, since $p$ is odd and $k < 4p^\ell$, we have $|c|\in\{0,2p^\ell\}$.
Hence, if $|c|\not\in\{0,2p^\ell\}$, then $|c|\not\equiv 0\bmod p^\ell$.

Recall from Definition \ref{def:H_k} that  $H_k$ can be defined as the graph whose vertices are the strings of $\{0,1\}^k$ with an even Hamming weight and where two distinct vertices $u,v$ are adjacent if their Hamming distance $|u\oplus v|$ is equal to $(k+1)/2 = 2p^\ell$. Here $u\oplus v$ is the sum modulo 2.
For $u,v\in V(H_k)$, their Hamming distance $|u\oplus v|$ is an even number. Hence if $u\ne v$ are not adjacent in $H_k$, then $|u\oplus v|\not\in \{0,2p^\ell\}$ and thus $|u\oplus v|\not\equiv 0 \mod p^\ell$.

Let $f\in\Z_p[x_1,\dots,x_k]$ be a multilinear polynomial of degree at most $p^\ell-1$ such that for every $c\in \bset{k}$, we have
\beqn
f(c)\equiv \left\{
\begin{array}{ll}
1 & \text{if } |c|\equiv 0\bmod p^\ell\\
0 & \text{otherwise,}
\end{array}
\right.
\eeqn
as is promised to exist by Lemma~\ref{lem:bbr} (applied to $w=0$).

We use $f$ to define a representation for~$H_k$.
To this end define for each $u\in\bset{k}$ vertex in $V(H_k)$ the polynomial $f_u\in \Z_p[x_1,\dots,x_k]$ obtained by replacing in the polynomial $f$ the variable $x_i$ by $1-x_i$ if $u_i = 1$ and leaving it unchanged otherwise.
For example, if $u = (1,1,0,\dots,0)$, then  $f_u(x_1,\dots,x_k) = f(1-x_1,1-x_2,x_3,\dots,x_k)$.
Moreover, associate to the vertex $u$  the point $c_u=u$ seen as a $0/1$ vector in~$\Z_p^k$.
We claim that $\big((f_u,c_u)\big)_{u\in V(H_k)}$ is a representation of $H_k$. 
To see this, observe that  $f_u(c_v) = f(u\oplus v)$ for any $u,v\in V(H_k)$, so that $f_u(c_u)=f(0)=1$, and $f_u(c_v)=0$ if $u,v$ are distinct and non-adjacent.

Since the polynomials $f_u$ are multilinear and have degree at most $p^\ell-1$, they span a space of dimension at most ${{k\choose 0}  + {k\choose 1} + \cdots + {k\choose p^\ell-1}},$ which is the number of multilinear monomials of degree at most~$p^\ell -1$.
Applying Lemma~\ref{lem:polybound} we obtain that
\begin{equation}\label{eqa}
\alpha(H_k^{\boxtimes m})\leq \left( {k\choose 0}  + {k\choose 1} + \cdots + {k\choose p^\ell-1}\right)^m.
\end{equation}
We now use the well known fact that for $q,k\in \N$ with ${1<q<k/2}$, 
${k \choose 0}+\ldots + {k\choose q-1}\le 2^{kH(q/k)}$.
From this, since 
$p^\ell / (4p^\ell -1) \le 3/11$, we deduce that the right hand side in (\ref{eqa}) can be upper bounded by
$2^{k\, m\, H(3/11)} < 2^{0.846\, k \, m}$.
\end{proof}

The bound~\eqref{lem:witsbound-2} stated in Theorem~\ref{thm:wits} is a simple corollary of Lemma~\ref{lem:newshan}.

\begin{corollary}
Let $p$ be an odd prime number and~$\ell\in\N$. Then, for $k = 4p^{\ell} - 1$, we have $R(H_k) \geq 0.154k - 1$.
\end{corollary}

\begin{proof}
By Lemma~\ref{lem:newshan}, for every integer $m$ we have
\beqn
\chi(H_k^{\boxtimes m}) \geq \frac{|V(H_k^{\boxtimes m})|}{\alpha(H_k^{\boxtimes m})} > \frac{2^{(k-1)m}}{2^{0.846km} }= 2^{(0.154k-1)m}.
\eeqn 
Taking the logarithm, dividing by $m$ and taking the limit~${m\to\infty}$ gives the result.
\end{proof}

\section{Separation between classical and entangled Shannon capacity}\label{sec:Shannon}
Here we prove Theorem~\ref{thm:cap}, thus showing the existence of an infinite family of graphs for which the entangled capacity exceeds the Shannon capacity.
We repeat the statement of the theorem here for convenience.

\begin{reptheorem}{thm:cap}
For every odd integer~$k \geq 11$, we have 
\begin{equation}\label{lem:qcapbound-2}
\qcap(H_k) \geq (k-1)\left(1-\frac{2\log(k+1)}{k-3}\right).
\end{equation}
Moreover, if $k=4p^\ell -1$ where $p$ is an odd prime and $\ell\in \N$, then 
\begin{equation}\label{lem:newshan-cap-2}
c(H_k) \leq 0.846\, k.
\end{equation}
\end{reptheorem}

\subsection{Lower bound on the entangled Shannon capacity}

The proof of the bound~(\ref{lem:qcapbound-2}) on the entangled Shannon capacity is again based on remote state preparation (Section~\ref{sec:rsp}).

The idea behind the proof is to show that with $t+1$ sequential uses of a channel with confusability graph~$H_k$, Alice can send Bob~$|V(H_k)|^t$ distinct messages with zero probability of error, provided that $t\leq \log\alpha(H_k)/\log \xi'(H_k)$.
More generally, we have the following lemma, from which the result~\eqref{lem:qcapbound-2} easily follows.

\begin{lemma}\label{lem:rsp-channel}
For every graph $G$ and integer $t \geq 1$ such that $t \leq  \log \alpha(G)/\log \xi'(G)$, we have
$$ \qcap(G) \geq \frac{\log|V(G)|}{1+ \frac{1}{t}}. $$
\end{lemma}

\begin{proof}
Let~$\mathcal N = (\mathsf S, \mathsf V, Q)$ be a channel with confusability graph~$G$.
Let~$d= \xi'(G)$ and let~$f$ be a $d$-dimensional orthogonal representation of $G$ such that its vectors have entries of modulus one. 
For each~$x\in V(G)$ define~$\rho_x = f(x)f(x)^*/d$.
Let~$t$ be a positive integer such that
\beq\label{eq:tbound}
t \leq \frac{\log \alpha(G)}{\log d}.
\eeq
By~\cite[Theorem~1]{Cubitt:2010} it suffices to find an entanglement-assisted protocol  for the noiseless transmission of~$|V(G)|^t$ distinct messages based on at most~$t+1$ uses of the channel~$\mathcal N$.
Indeed this then implies that~$\qindep(G^{\boxtimes (t+1)}) \geq |V(G)|^t$ and therefore
\beqn
\qcap(G) \geq \frac{\log\qindep(G^{\boxtimes (t+1)})}{t+1} \geq \frac{t\log |V(G)|}{1+t}
\eeqn
as claimed.
To this end, we consider the following four-step protocol for transmitting a sequence ${{\bf x} = (x_1,\dots,x_t)\in V(G)^t}$. 
First, Alice prepares $d$-dimensional quantum registers~$\mathcal A_1,\dots,\mathcal A_t$ to be in the states $\rho_{x_1},\dots,\rho_{x_t}$, respectively.
Second, Alice sends the sequence~${\bf x}$ through the channel by using it $t$ times in a row.
This will result in~$t$ channel-outputs on Bob's end of the channel from which he can infer that each~$x_i$ belongs to a particular clique in~$G$.
Third, Alice and Bob execute the remote state preparation scheme described in Section~\ref{sec:rsp} $t$ times in a row, once for each of the states $\rho_{x_1},\dots,\rho_{x_t}$ separately. (Recall that $\rho_{x_i} = f(x_i)f(x_i)^*/d$ where $f(x_i) \in \C^d$ has norm~$\sqrt{d}$, so in the notation of Section~\ref{sec:rsp} we are setting~$u = f(x_i)/\sqrt{d}$.)
This requires that Alice communicates a total of $t\ceil{\log d}$ bits to Bob.
To do so, Alice uses the channel one more time to send, without error, the bits required to perform the remote state preparation. This can be done if $\log \alpha(G) \geq t \ceil{\log d}$, which holds by our assumed bound~\eqref{eq:tbound}.
At this point Bob's quantum registers~$\mathcal B_1,\dots,\mathcal B_t$ are in states $\rho_{x_1},\dots,\rho_{x_t}$. 
Moreover, for each $x_i$ Bob knows a clique in the graph $G$ that contains~$x_i$ and by construction elements of a clique have pairwise orthogonal states. 
In the last step, for every $i \in [t]$, Bob can perform a measurement on register $\mathcal B_i$ such that he gets outcome~$x_i$ with probability one (due to Lemma~\ref{lem:helho}). 
Hence, Bob can recover any sequence~${\bf x} = (x_1,\dots,x_t)\in V(G)^t$ with zero probability of error, completing the proof.
\end{proof}

\begin{proof}[Proof of~\eqref{lem:qcapbound-2}]
Recall from the proof of Lemma~\ref{lem:qchrom} that $\log \xi'(H_k) \leq \log (k+1)$.
From  Lemma~\ref{lem:alphaHk}, we know that $\log \alpha(H_k)\geq (k-3)/2$.
Hence, for $k \geq 11$ we can choose $t = \left\lfloor (k-3)/(2 \log (k+1))\right\rfloor \leq \left\lfloor \log \alpha(H_k)/\log \xi'(H_k) \right\rfloor$. Note that we require $k \geq 11$ to ensure that $t \geq 1$.  
We now apply Lemma~\ref{lem:rsp-channel} and, recalling that $|V(H_k)| = 2^{(k-1)}$, we obtain 
\beqrn
\qcap(H_k) \geq (k-1) \frac{1}{1+\frac{1}{t}} \geq (k-1) \left(1-\frac{2\log(k+1)}{k-3}\right)
\eeqrn
which gives the result.
\end{proof}

\subsection{Upper bound on the Shannon capacity}

The upper bound (\ref{lem:newshan-cap-2}) on the Shannon capacity of $H_k$ (for certain values of $k$) stated in Theorem~\ref{thm:cap} is an easy corollary of Lemma~\ref{lem:newshan}.

\begin{corollary}
Let~$p$ be an odd prime, $\ell\in\N$ and set ${k = 4p^{\ell} - 1}$. Then, $c(H_k) \leq 0.846k$
\end{corollary}

\begin{proof}
By taking the logarithm, dividing by $m$ and taking the limit~$m\to\infty$ on both sides of~\eqref{eq:newshan} we get the result.
\end{proof}

\section{Separation between classical and entangled source-channel cost rate}\label{sec:scratesep}

\subsection{Proof of Theorem \ref{thm:scratesep}}

Now we prove Theorem~\ref{thm:scratesep}, separately showing the two bounds (\ref{eq:qscratebound}) for $\qscrate$ and (\ref{eq:scratebound}) for
$\scrate$.

\begin{reptheorem}{thm:scratesep}
Let~$p$ be an odd prime and $\ell\in\N$ such that there exists a Hadamard matrix of size~$4p^{\ell}$.
Set $k = 4p^\ell - 1$.
Then, 
\beq\label{eq:qscratebound-2}
\qscrate(H_k,H_k) \leq \frac{\log (k+1)}{(k-1)\left(1-\frac{2\log(k+1)}{k-3}\right)
},
\eeq
\beq\label{eq:scratebound-2}
\scrate(H_k,H_k) >  \frac{0.154\, k-1} {k-1-\log(k+1)}.
\eeq
\end{reptheorem}

The bound~(\ref{eq:qscratebound-2}) is obtained by combining~(\ref{lem:qwitsbound}),~(\ref{lem:qcapbound}) with Proposition~\ref{prop:eta}.
The proof of~(\ref{eq:scratebound}) relies on the No-Homomorphism Lemma due to Albertson and Collins~\cite{Albertson:1985}.

An {\em automorphism} of $G$ is a permutation $\pi$ of $V(G)$ preserving edges, i.e., $\{\pi(u),\pi(v)\}\in E(G)$  if and only if $\{u,v\}\in E(G)$.
The graph $G$ is {\em vertex-transitive} if, for any $u,v\in V(G)$, there exists an automorphism $\pi$ of $G$ such that $v=\pi(u)$.

\begin{lemma}[No-Homomorphism Lemma, Albertson and Collins~\cite{Albertson:1985}]\label{lem:transbound} 
Let $H$ be a vertex-transitive graph. If there is a homomorphism from $G$ to $H$, then
\beqn
\frac{|V(G)|}{\alpha(G)} \leq  \frac{|V(H)|}{\alpha(H)}. 
\eeqn
\end{lemma}

As observed in \cite{Briet:2012c}, the graph $H_k$ is vertex-transitive; indeed, for any $u\in V(H_k)$, the map $v\mapsto u\oplus v$ is an automorphism of $H_k$.
It is easy to see that vertex transitivity is preserved under strong products and complements. Hence, $\overline {H_k^{\boxtimes n}}$ is vertex-transitive for any $n\in \N$.

\begin{proof}[Proof of~\eqref{eq:scratebound-2}]
Recall the definition of  $\scrate(H_k,H_k)$ from (\ref{eqeta}).
Consider integers $m,n\in \N$ for which  $H_k^{\boxtimes m}\longrightarrow \overline{H_k^{\boxtimes n}}$.
Applying Lemma \ref{lem:transbound}, we deduce that 
\beq\label{eq:homineq}
\frac{\lvert V(H_k^{\boxtimes m}) \lvert}{\alpha(H_k^{\boxtimes m})} 
\leq 
\frac{\lvert V(\overline{H_k^{\boxtimes n}}) \lvert}{\alpha(\overline{H_k^{\boxtimes n}})} 
= 
\frac{\lvert V(\overline{H_k^{\boxtimes n}}) \lvert}{\omega(H_k^{\boxtimes n})}.
\eeq
From Corollary~\ref{cor:clique} we have $\omega(H_k^{\boxtimes n}) = (k+1)^n$.
As $|V(H_k)| = 2^{k-1}$ and applying Lemma~\ref{lem:newshan}, we get
\beqn
\frac{2^{(k-1)\, m}}{2^{k\, m\, 0.846}} 
\:\:\stackrel{\text{Lemma \ref{lem:newshan}}}{<} \:\:
\frac{\lvert V(H_k^{\boxtimes m}) \lvert}{\alpha(H_k^{\boxtimes m})}
\:\:\stackrel{\eqref{eq:homineq}}{\leq} \:\:
\frac{\lvert V(\overline{H_k^{\boxtimes n}}) \lvert}{\omega(H_k^{\boxtimes n})}
\:\:=\:\:
\frac{2^{(k-1)\, n}}{(k+1)^n}\,.
\eeqn
After a few elementary algebraic manipulations and taking logarithms the above inequality implies 
$$\frac{n}{m} >  \frac{0.154\, k-1} {k-1-\log(k+1)}.$$
\end{proof}

\section{Concluding remarks and open problems}\label{sec:concl}

We have shown a separation between classical and entangled-assisted coding for the zero-error source-channel, source and channel problems. 
Note that these separations do not hold if asymptotically vanishing error is allowed.
We have presented an infinite family of instances for which there is an exponential saving in the minimum asymptotic cost rate of communication for the source-channel and the source coding problems.  
Moreover, for the channel coding problem we showed an infinite family of channels for which the entangled Shannon capacity exceeds the classical Shannon capacity by a constant factor. 
It would be interesting to find a family of channels with a larger separation.

The main result in~\cite{Nayak:2006} is that, for the classical source-channel coding problem, there exist situations for which separate encoding is highly suboptimal. Does this happen also in the entanglement-assisted case?
This question has a positive answer if there exists a graph $G$ with $\qwits(G) > \qcap(\overline G)$.
In~\cite{Nayak:2006} a sufficient condition for a separate encoding to be optimal is also proven, namely that the characteristic or the confusability graph is a perfect graph.
It is straightforward to see that this is also a sufficient condition for a separate entangled-assisted encoding to be optimal. Are there weaker conditions that hold for the entangled case?

One of the most interesting open questions in zero-error classical information theory is the computational complexity of the Witsenhausen rate and of the Shannon capacity. The same question is also open for the entangled counterparts as well as for the parameters $\qchrom$ and $\qindep$.

In Section~\ref{sec:intro}, we have seen that the entangled chromatic and independence number generalize the parameters $\chi_q$ and $\alpha_q$ which arise in the context of Bell inequalities and non-local games.
In~\cite{Roberson:2012} it is conjectured that $\qindep(G)=\alpha_q(G)$ for every graph $G$.
A possible approach to show that $\qchrom$ and $ \chi_q$ are two separate quantities is to prove that the relationship between Kochen-Specker sets and $\chi_q$ found in~\cite{Mancinska:2012} does not hold for $\qchrom$.
Finally, we mention that the existence of a graph $G$ for which $\qchrom(G)< \chi_q(G)$ or $\alpha_q(G) < \qindep(G)$ would prove the existence of a non-local game such that every quantum strategy that wins with probability one does not use a maximally entangled state.

\section*{Acknowledgments}
We thank Ronald de Wolf for useful discussions and helpful comments on a preliminary version of this paper, and we thank Dion Gijswijt and Debbie Leung for helpful discussions, and the anonymous referees for helpful comments on an earlier version of this manuscript, for pointing out an error in a previous version of Lemma~\ref{lem:helho}, and for providing the reference for Theorem~\ref{lemma:hjw}.

J.~B.\ was partially supported by a Rubicon grant from the Netherlands Organization for Scientific Research (NWO).
J.~B.\ and H.~B.\ were supported by the European Commission under the project QCS (Grant No. 255961). 
H.~B.\ and T.~P.\ were supported by the EU grant SIQS.
G.~S.\ was supported by Ronald de Wolf's Vidi grant 639.072.803 from the Netherlands Organization for Scientific Research (NWO). This work was done while G.S.
was a PhD student at CWI, Amsterdam.

\end{document}